\documentclass[11pt]{article}

\usepackage{microtype}
\usepackage{mathtools,mathrsfs}
\usepackage[table,svgnames]{xcolor}
\usepackage[colorlinks=true,linkcolor=black,citecolor=MidnightBlue,urlcolor=MidnightBlue]{hyperref}
\usepackage{xspace}
\usepackage{tikz}
\usepackage{fullpage}

\usepackage{amssymb,amsmath,amsfonts,fullpage,mdwlist,enumerate,amsthm, graphicx,algpseudocode, algorithm, algorithmicx, cite}
\usepackage{thmtools, thm-restate}

\newtheorem{theorem}{Theorem}[section]
\newtheorem{lemma}[theorem]{Lemma}
\newtheorem{observation}{Observation}

\newtheorem{corollary}[theorem]{Corollary}

\newtheorem{claim}[theorem]{Claim}
\newtheorem{definition}{Definition}[section]

\newtheorem{property}{Property}[section]

\title{Distributed Approximation of Minimum $k$-edge-connected\\ Spanning Subgraphs}
\author{Michal Dory\footnote{Technion, Department of Computer Science, \texttt{smichald@cs.technion.ac.il}. Supported in part by the Israel Science Foundation (grant 1696/14).}
}

\newcommand{\congest}{\textsc{Congest}\xspace}

\newcommand{\remove}[1]{}
\DeclareMathOperator*{\Moplus}{\text{\raisebox{0.25ex}{\scalebox{0.8}{$\bigoplus$}}}}

\remove{
\[
\Moplus_{i=1}^{k}V_i \qquad A\Moplus B_{A \Moplus B_{A \Moplus B}}
\]
}

\begin{document}

\begin{titlepage}

\maketitle

\begin{abstract}
In the minimum $k$-edge-connected spanning subgraph ($k$-ECSS) problem the goal is to find the minimum weight subgraph resistant to up to $k-1$ edge failures. This is a central problem in network design, and a natural generalization of the minimum spanning tree (MST) problem. While the MST problem has been studied extensively by the distributed computing community, for $k \geq 2$ less is known in the distributed setting.

In this paper, we present fast randomized distributed approximation algorithms for $k$-ECSS in the \congest model.
Our first contribution is an $\widetilde{O}(D + \sqrt{n})$-round $O(\log{n})$-approximation for 2-ECSS, for a graph with $n$ vertices and diameter $D$. The time complexity of our algorithm is almost tight and almost matches the time complexity of the MST problem. For larger constant values of $k$ we give an $\widetilde{O}(n)$-round $O(\log{n})$-approximation. Additionally, in the special case of \emph{unweighted} 3-ECSS we show how to improve the time complexity to $O(D \log^3{n})$ rounds. All our results significantly improve the time complexity of previous algorithms.
\end{abstract}



 
\thispagestyle{empty}
\end{titlepage}

\section{Introduction}
The edge-connectivity of a graph determines its resistance to edge failures, which is crucial for network reliability.   
In the minimum weight $k$-edge-connected\footnote{An undirected graph $G$ is \emph{$k$-edge-connected} if it remains connected after the removal of any $k-1$ edges.} spanning subgraph ($k$-ECSS) problem the input is a $k$-edge-connected graph $G$, and the goal is to find the minimum weight $k$-ECSS of $G$. 
The minimum $k$-ECSS problem is widely studied in the sequential setting (see, e.g., \cite{jain2001factor, khuller1994biconnectivity,cheriyan2000approximating, gabow2009approximating,gabow2012iterated,goemans1994improved}).
The unweighted version of the problem admits an $1+\frac{1}{2k}+O(\frac{1}{k^2})$ approximation \cite{gabow2012iterated}, and the weighted problem admits 2-approximations \cite{jain2001factor, khuller1994biconnectivity}. Many additional related connectivity problems are studied in the sequential setting, see \cite{khuller1996approximation, kortsarz2010approximating} for surveys.

However, because of the distributed nature of networks, it is crucial to study the problem also from the distributed perspective. The 1-ECSS problem is just the minimum spanning tree (MST) problem, which is a central and well-studied problem in the distributed setting (see, e.g., \cite{gallager1983distributed,garay1998sublinear,kutten1998fast,elkin2006unconditional, DBLP:conf/podc/Elkin17, pandurangan2017time}). In the \congest model, there is an $O(D + \sqrt{n} \log^*{n})$-round algorithm for the problem \cite{kutten1998fast} for a graph with $n$ vertices and diameter $D$, which is almost tight due to an $\Omega(D + \sqrt{\frac{n}{\log{n}}})$ lower bound \cite{peleg2000near, elkin2006unconditional, sarma2012distributed}.
Although an MST is a sparse low-cost backbone of a graph, even a single edge failure disconnects it and completely destroys its functionality. Hence, it is crucial to find low-cost backbones with higher connectivity.

Yet, for $k>1$ less is known in the distributed setting. For \emph{unweighted} 2-ECSS there is an $O(D)$-round 2-approximation \cite{censor2017fast}, and an $O(n)$-round $\frac{3}{2}$-approximation \cite{krumke2007distributed}. For \emph{unweighted} $k$-ECSS, there is an $O(k(D+\sqrt{n}\log^*{n}))$-round 2-approximation \cite{thurimella1995sub}. This algorithm is based on repeatedly finding maximal spanning forests in the graph, and removing their edges from the graph, 
which results in a $k$-ECSS with at most $k(n-1)$ edges. This guarantees a 2-approximation for the \emph{unweighted} case, since any $k$-ECSS has at least $\frac{kn}{2}$ edges.
However, this approach cannot extend to the weighted case, since in the weighted case even adding one redundant edge may be too expensive.

A natural question is whether it is possible to design efficient approximations also for \emph{weighted} $k$-ECSS. 
For weighted 2-ECSS there are $O(n\log{n})$-round \cite{krumke2007distributed} and $O(h_{MST}+\sqrt{n}\log^*{n})$-round \cite{censor2017fast} 3-approximations, where $h_{MST}$ is the height of the MST of the graph. Both these algorithms start by building an MST and then augment it to be 2-edge-connected. To do so, they use algorithms for the weighted \emph{tree augmentation problem} (TAP), in which the goal is to augment the connectivity of a given spanning tree $T$ to 2 by adding a minimum cost set of edges from the graph $G$ to $T$. 
However, currently the best algorithm for solving weighted TAP takes $O(h)$ rounds \cite{censor2017fast}, where $h$ is the height of $T$. Since the algorithm for 2-ECSS augments an MST, it results in a time complexity that depends on $h_{MST}$, which can be $\Theta(n)$ in the worst case.

To the best of our knowledge, the only distributed algorithm for weighted $k$-ECSS for $k>2$ is an $O(knD)$-round $O(\log{k})$-approximation algorithm \cite{shadeh2009distributed} based on a primal-dual algorithm of Goemans et al. \cite{goemans1994improved}, that solves even the more general Steiner Network problem. If $k$ is constant and $D$ is small, the time complexity of \cite{shadeh2009distributed} is close to linear, but in the worst case the time complexity is $\Omega(n^2)$, which is trivial in the distributed setting.\footnote{All previous results use deterministic algorithms.} 
   
In this paper, we address this fundamental topic, and present efficient distributed approximation algorithms for $k$-ECSS. Our algorithms work in the \congest model of distributed computing \cite{peleg2000distributed}, in which vertices exchange messages of $O(\log{n})$ bits in synchronous rounds. 

\subsection{Our contributions}

Our first contribution is the first sublinear algorithm for weighted 2-ECSS. 

\begin{restatable}{theorem}{twoECSS} \label{2-ECSS}
There is a distributed algorithm for weighted 2-ECSS in the \congest model that guarantees an approximation ratio of $O(\log{n})$, and takes $O((D + \sqrt{n})\log^2{n})$ rounds w.h.p.
\end{restatable}

The time complexity of our algorithm improves upon the previos $O(h_{MST}+\sqrt{n}\log^*{n})$-round algorithm \cite{censor2017fast}, it almost matches the time complexity of the MST problem, and it is almost tight. 
Computing an $\alpha$-approximation for weighted 2-ECSS requires $\Omega(D + \frac{\sqrt{n}}{\log{n}})$ rounds, for any polynomial function $\alpha$ \cite{censor2017fast}.\footnote{The lower bound in \cite{censor2017fast} is for weighted TAP, however an $\alpha$-approximation algorithm for weighted 2-ECSS gives an $\alpha$-approximation algorithm for weighted TAP where we give to the edges of the input tree $T$ weight 0. Hence, a lower bound for weighted TAP implies a lower bound for weighted 2-ECSS.} 
We next consider the case $k>2$, and show the following.

\begin{restatable}{theorem}{kECSS}
There is a distributed algorithm for weighted $k$-ECSS in the \congest model with an expected approximation ratio of $O(k \log{n})$, and time complexity of $O(k(D \log^3{n} + n))$ rounds. 
\end{restatable}

This gives the first nearly-linear time algorithm for any constant $k$, and improves upon the previous $O(knD)$-round algorithm \cite{shadeh2009distributed}. We also show that in the special case of \emph{unweighted} 3-ECSS we can improve the time complexity of the algorithm to $O(D \log^3{n})$ rounds, which improves upon the previous $O(D+\sqrt{n}\log^*{n})$-round algorithm \cite{thurimella1995sub}.

\begin{restatable}{theorem}{threeECSS}
There is a distributed algorithm for unweighted $3$-ECSS in the \congest model with an expected approximation ratio of $O(\log{n})$, and time complexity of $O(D \log^3{n})$ rounds.
\end{restatable}

\subsection{Additional related work} 

\textbf{FT-MST:} In the fault-tolerant MST problem the goal is to find a sparse subgraph of the input graph $G$ that contains an MST of $G \setminus\{e\}$ for each edge $e$.
This problem can be solved in $O(D + \sqrt{n} \log{n})$ rounds using the distributed algorithm of Ghaffari and Parter \cite{ghaffari2016near}. 
One of the ingredients in \cite{ghaffari2016near} is a decomposition of the tree into segments that turns out to be useful also for our 2-ECSS algorithm.
While a FT-MST and a minimum 2-ECSS are both low-cost 2-edge-connected spanning subgraphs, the main difference between the problems is that in the latter the goal is to minimize the \emph{sum of costs} of edges in the solution. Hence, the total cost of a solution for minimum 2-ECSS may be much cheaper. \\[-7pt]

\textbf{Cycle space sampling:} The \emph{cycle space sampling} technique introduced by Pritchard and Thurimella \cite{pritchard2011fast} allows to detect small cuts in a graph using connections between the cycles and cuts in a graph. They show an $O(D)$-round algorithm that assigns the edges of a graph short labels that allow to detect cuts of size 1 or 2 in $O(D)$ rounds. In particular, this gives an $O(D)$-round algorithm for verifying if a graph is 2-edge-connected or 3-edge-connected (for more details, see Section \ref{3-full}).
We use this technique to show an efficient algorithm for the minimum \emph{size} 3-ECSS. \\[-7pt]

\textbf{Additional related problems:} 
There are also other $O(D)$-round algorithms for verifying if a graph is 2-edge-connected \cite{pritchard2005robust, censor2017fast}. A natural approach for verifying if a graph is $k$-edge-connected is computing the size $\lambda$ of a minimum cut in a graph. There are approximation and exact algorithms for the minimum cut problem in $\widetilde{O}(D + \sqrt{n})$ rounds \cite{DBLP:journals/corr/abs-1305-5520, nanongkai2014almost} and  $\widetilde{O}((D + \sqrt{n})\lambda^4)$ rounds \cite{nanongkai2014almost}, respectively, and a lower bound of  $\Omega(D + \frac{1}{\log{n}} \sqrt{\frac{n}{\alpha \lambda}})$ rounds for an $\alpha$-approximation in a graph with diameter $D = O(\log{n} + \frac{1}{\lambda \log{n}}\sqrt{\frac{n}{\alpha \lambda}})$ \cite{DBLP:journals/corr/abs-1305-5520} (note that if $\lambda = O(1)$ this lower bound becomes $\Omega(D)$).
Another related problem studied in the distributed setting is the decomposition of a graph with large connectivity into many disjoint trees, while almost preserving the total connectivity through the trees \cite{censor2014distributed}. \\[-7pt]

\textbf{Covering problems:} We show that the minimum $k$-ECSS problem is closely related to the set cover problem.  Some elements in our algorithms and analysis are inspired by the parallel set cover algorithm of Rajagopalan and Vazirani~\cite{rajagopalan1998primal}, and the minimum dominating set (MDS) algorithm of Jia et al.~\cite{jia2002efficient}. It is worth noting that our 2-ECSS algorithm \emph{guarantees} the approximation ratio, where other distributed $O(\log{n})$-approximations for set cover obtain approximations that hold in expectation or w.h.p \cite{jia2002efficient,kuhn2016local}. For this reason, we believe that our approach can be useful for additional local or global covering problems, particularly in scenarios where it is important to guarantee the approximation. We also used it in a recent algorithm for the minimum 2-spanner problem \cite{spanner}.
 
\subsection{Preliminaries} \label{sec:pre}
 
\textbf{Problem definition:} An undirected graph $G$ is \emph{$k$-edge-connected} if it remains connected after the removal of any $k-1$ edges. 
In the \emph{minimum weight $k$-edge-connected spanning subgraph problem ($k$-ECSS)} the input is an undirected $k$-edge-connected graph $G$ with $n$ vertices and non-negative weights $w(e)$ on the edges. The goal is to find the minimum weight $k$-ECSS of $G$. We assume that the weights of the edges are integers and are polynomial in $n$. This guarantees that a weight can be represented in $O(\log n)$ bits. 
In \emph{unweighted} $k$-ECSS the goal is to find the minimum \emph{size} $k$-ECSS of $G$. It is equivalent to weighted $k$-ECSS where all the edges have unit-weight. \\[-7pt]

\textbf{The distributed model:} In the distributed \congest model, the input graph $G$ to the problem is the communication network. Initially all the vertices know the ids of their neighbors and the weights of the edges adjacent to them, and at the end each vertex knows which of the edges adjacent to it are taken to the solution. In our algorithms, it would be convenient to say that the edges do some computations. When we say this, we mean that the endpoints of the edge simulate the computation. \\[-7pt]

During our algorithm it is useful to communicate over a BFS tree. We construct a BFS tree with root $r$ in $O(D)$ rounds \cite{peleg2000distributed}, where $r$ is the vertex with minimum id. Using the BFS tree we can distribute $\ell$ different messages from vertices in the tree to all the vertices in the tree in $O(D + \ell)$ rounds using standard techniques \cite{peleg2000distributed}. We assume that vertices know the number of vertices $n$ and the number of edges $m$ in the algorithm, they can learn this information in $O(D)$ rounds by communication over the BFS tree.
For a rooted tree $T$, we denote by $p(v)$ the parent of $v$ in the tree, and we denote by $LCA(u,v)$ the lowest common ancestor of $u$ and $v$. \\[-7pt]


\textbf{Roadmap:} The rest of the paper is organized as follows. In Section \ref{tech}, we give a high-level overview of our algorithms. 
In Sections \ref{sec:2-full}, \ref{k-full} and \ref{3-full} we present our 2-ECSS, $k$-ECSS and 3-ECSS algorithms, respectively. Finally, in Section \ref{disc}, we discuss questions for future research.

\section{Technical overview} \label{tech}

A natural approach for finding a minimum $k$-ECSS is to start with an empty subgraph $H$, and in iteration $i$ for $1 \leq i \leq k$ augment its connectivity from $i-1$ to $i$. We define the problem $Aug_k$ as follows. Given a $k$-edge-connected graph $G$ and a $(k-1)$-edge-connected spanning subgraph $H$ of $G$, the goal is to find a minimum weight set of edges $A$ from $G$, such that $H \cup A$ is $k$-edge-connected. For a set of edges $H$, we define $w(H) = \sum_{e \in H} w(e).$ We next show the following.

\begin{restatable}{claim}{Augk} \label{Aug_k} 
Assume that for $1 \leq i \leq k$, the algorithm $A_i$ is an $\alpha_i$-approximation algorithm for $Aug_i$ that takes $T_i$ rounds, then there is an $(\sum_{i=1}^{k} \alpha_i)$-approximation algorithm for $k$-ECSS that takes $\sum_{i=1}^{k} T_i$ rounds.
\end{restatable}

\begin{proof}
The algorithm for $k$-ECSS starts with an empty subgraph, and in iteration $i$ uses the algorithm $A_i$ to augment the connectivity from $i-1$ to $i$. Let $H_i$ be the set of edges added to the augmentation in iteration $i$, and let $H = \bigcup_{i=1}^{k} H_i$ be the solution constructed. $H$ is clearly $k$-edge-connected.
The correctness follows from the fact that an optimal solution $H^*$ for $k$-ECSS is a set of edges that augments the connectivity of any subgraph produced in the algorithm to $k$ (and in particular to any $i \leq k$). Since the algorithm $A_i$ is an $\alpha_i$-approximation algorithm, this gives $w(H_i) \leq \alpha_i w(H^*)$. This shows that $w(H) = \sum_{i=1}^{k} w(H_i) \leq \sum_{i=1}^{k} \alpha_i w(H^*)$, as needed. The time complexity of the algorithm is $\sum_{i=1}^{k} T_i$ rounds.
\end{proof}

A set of edges $C$ is a \emph{cut} in a connected graph $G$ if $G \setminus C$ is disconnected. 
Let $H$ be a subgraph of a graph $G=(V,E)$, and let $C$ be a cut in $H$. 

\begin{definition}
An edge $e \in E$ \emph{covers} the cut $C$ if $(H \setminus C) \cup \{e\}$ is connected.
\end{definition}

Note that in a $(k-1)$-edge-connected graph $H$, the minimum cut is of size at least $k-1$. To solve $Aug_k$ our goal is to find a minimum cost set of edges $A$ that covers all the cuts of size $k-1$ in $H$. This is a special case of the set cover problem. 

\subsection{General framework of our algorithms} \label{sec:frame}

In order to solve $k'$-ECSS, we present an algorithm for $Aug_k$ for any constant $k \leq k'$. In the algorithm we maintain a set of edges $A$ that contains all the edges added to the augmentation. 
For an edge $e \not \in H$, we denote by $S_e$ the set of cuts of size $k-1$ of $H$ that $e$ covers. During the algorithm and analysis, 
we denote by $C_e$ all the cuts in $S_e$ that are still not covered by edges added to $A$. \\[-7pt]

\textbf{Cost-effectiveness.}
The \emph{cost-effectiveness} of an edge $e \not \in H$ is $\rho(e) = \frac{|C_e|}{w(e)}$. 
The \textit{rounded} cost-effectiveness of an edge $e$, denoted by $\tilde{\rho}(e)$, is obtained by rounding $\rho(e)$ to the closest power of 2 that is greater than $\rho(e)$. If $w(e) = 0$, the values $\rho(e)$ and $\tilde{\rho}(e)$ are defined to be $\infty$.
Adding an edge with maximum value of cost-effectiveness to $A$ allows to cover many cuts, while paying minimal cost. This suggests the following \emph{sequential} greedy algorithm. At each step, we add to $A$ the edge with maximum cost-effectiveness, and we continue until all the cuts of size $k-1$ are covered. This approach that is based on the classic greedy algorithm for set cover achieves an $O(\log{n})$-approximation. \\[-7pt]

\textbf{Symmetry breaking.}
Adding only one edge in each step gives an algorithm that is inherently sequential, and in order to obtain an efficient \emph{distributed} algorithm we would like to add many edges to $A$ in parallel. A naive approach could be to add all the edges with maximum cost-effectiveness to $A$ simultaneously. However, then we may add too many edges to $A$ and the approximation ratio is no longer guaranteed. To overcome this, we consider all the edges with maximum rounded cost-effectiveness as \emph{candidates}, and then we break the symmetry between the candidates. After that, some of the candidates are added to $A$, and we proceed in iterations until all the cuts of size $k-1$ are covered.
This gives us the general structure of the algorithm. \\[-7pt]

All our algorithms have this structure, but they differ in the symmetry breaking mechanism and by the implementation of the cost-effectiveness computation and other computations required. We next overview the main techniques in each of our algorithms.

\subsection{2-ECSS} \label{sec:2-brief}

In the special case of $k=2$, we start by building an MST, $T$, using the $O(D+ \sqrt{n}\log^*{n})$-round algorithm of Kutten and Peleg \cite{kutten1998fast}. Then, we augment $T$ to be 2-edge-connected using a new algorithm for weighted TAP. 
Although there is an algorithm with time complexity $\widetilde{O}(D+\sqrt{n})$ for \emph{unweighted} TAP \cite{censor2017fast}, this algorithm and its analysis rely heavily on the fact that the problem is unweighted, hence, giving an efficient algorithm for the weighted problem requires a different approach. 
Our algorithm for weighted TAP follows the framework described in Section \ref{sec:frame}, and exploits the simple structure of cuts of size 1 in the graph. 

As explained in Section \ref{sec:frame}, in order to augment the connectivity of $T$ to 2, our goal is to cover all the cuts of size 1 in $T$. Since a cut of size 1 is an edge, our goal is to cover all the tree edges. Let $e=\{u,v\}$ be a non-tree edge in $G$, and let $P_e$ be the unique path between $u$ and $v$ in $T$. Adding $e$ to $T$ creates a cycle that includes all the tree edges in $P_e$.
Since removing an edge from a cycle does not disconnect a graph, it follows that $e$ covers a tree edge $t$ if and only if $t \in P_e$. Hence, the set $S_e$ of all the cuts of size 1 covered by $e$ consists of all the tree edges of $P_e$, and $C_e$ are all the tree edges in $P_e$ that are still not covered by edges added to the augmentation $A$. 
In each iteration, all the non-tree edges with maximum rounded cost-effectiveness are the candidates. \\[-7pt]

\textbf{Symmetry breaking.}
In order to break the symmetry between the candidates, we suggest the following approach, inspired by a parallel algorithm for set cover \cite{rajagopalan1998primal}. We also used it recently in designing distributed algorithms for the 2-spanner and minimum dominating set (MDS) problems \cite{spanner}. Each candidate chooses a random number, and each tree edge chooses the first candidate that covers it according to the random values. A candidate edge $e$ that receives at least $\frac{|C_e|}{8}$ votes is added to $A$. Since we add to $A$ only edges receiving many votes, this approach allows us to add small number of edges to $A$ while covering many tree edges, and results in an $O(\log{n})$-approximation. \\[-7pt] 

\textbf{Efficient implementation.} 
A major difference between TAP and local covering problems like 2-spanner and MDS is that in the 2-spanner and MDS algorithms all the computations depend on a small local neighborhood around vertices, where TAP is a \emph{global} problem, and the algorithm requires to do many global computations simultaneously.
For example, during the algorithm each one of the non-tree edges needs to compute its cost-effectiveness and learn how many tree edges vote for it, and each one of the tree edges should vote for the first candidate that covers it.
However, there may be $\Theta(n^2)$ non-tree edges, and $n-1$ tree edges, and in order to get an efficient algorithm we should be able to do these computations in parallel.
To achieve this, we decompose the tree into segments with a relatively simple structure, following a decomposition used for solving the FT-MST problem \cite{ghaffari2016near}. The simple structure of the segments allows us to preform many computations simultaneously, which results in an $\widetilde{O}(D+\sqrt{n})$-round $O(\log{n})$-approximation for weighted TAP and for weighted 2-ECSS. 

\subsection{$k$-ECSS} \label{sec:k-ECSS}

The main obstacle in extending our 2-ECSS algorithm for larger values of $k$ is the need to work with many different cuts in parallel. For the case $k=2$, the cuts we needed to cover were represented by the tree edges. 
However, for larger values of $k$, we need to cover also cuts that contain several tree edges that may be far away from each other. Our goal is to design an algorithm for $Aug_k$ that follows the framework described in Section \ref{sec:frame}, but it is not clear anymore how to compute the cost-effectiveness of edges efficiently, and how to break the symmetry between candidates. \\[-7pt]

\textbf{Computing cost-effectiveness.} To compute the cost-effectiveness of edges, we use the following observation. The minimum $k$-ECSS has $O(kn)$ edges which is $O(n)$ for a constant $k$. If we guarantee that the subgraph $H$ constructed in our algorithm has $O(kn)$ edges, all the vertices can learn the complete structure of $H$ in $O(kn)$ time during the algorithm, and then each one of the (possibly $\Theta(n^2)$) edges can compute its cost-effectiveness locally. \\[-7pt]

\textbf{Symmetry breaking.} The main challenge is to break the symmetry between candidates. Let $deg(C)$ be the number of candidates that cover a cut $C$. To cover $C$, it is enough to add only one of these candidates to $A$, and if each of them is added to $A$ with probability $\frac{1}{deg(C)}$, then we add one candidate to cover $C$ in expectation. 
Ideally, we would like that each edge would be added to $A$ with probability that depends on the numbers $deg(C)$ of the cuts that it covers. A similar idea is used in the MDS algorithm of Jia et al. \cite{jia2002efficient}. However, in our case, it is not clear how to compute these values efficiently. 

In order to overcome this, we suggest the following ``guessing'' approach. Each candidate edge is added to $A$ with probability $p$. At the beginning $p = \frac{1}{m}$ where $m$ is the number of edges, after $O(\log{n})$ iterations we increase $p$ by a factor of $2$ and we continue until $p=1$. However, after each iteration, each candidate edge that was not added to $A$, computes its cost-effectiveness and remains a candidate only if its rounded cost-effectiveness is still maximal. The intuition is that the maximum degree of a cut decreases during the process. At the beginning $deg(C) \leq m$ for all cuts, however if there are cuts with degree close to $m$ they would probably be covered in the first $O(\log{n})$ iterations where $p = \frac{1}{m}$. Similarly, when we reach the phase that $p = \frac{1}{2^i}$, the maximum degree of a cut is at most $2^i$ w.h.p, and since each candidate is added to $A$ with probability $p = \frac{1}{2^i}$ then we do not add too many candidates to cover the same cuts. This allows us to prove an approximation ratio of $O(\log{n})$ in expectation. 
The number of iterations is $O(\log^3{n})$, since for each value of the $O(\log{n})$ possible values for rounded cost-effectiveness we have $O(\log^2{n})$ iterations.
Based on these ingredients we get an $O(D \log^3{n} + n)$-round $O(\log{n})$-approximation for $Aug_k$, which gives an $O(k(D \log^3{n} + n))$-round $O(k \log{n})$-approximation for $k$-ECSS.

\subsection{3-ECSS}

The bottleneck of our $k$-ECSS algorithm is the cost-effectiveness computation. In the special case of \emph{unweighted} 3-ECSS we show how to compute the cost-effectiveness in $O(D)$ rounds for a graph with diameter $D$. The main tool in our algorithm it the beautiful \emph{cycle space sampling} technique introduced by Pritchard and Thurimella \cite{pritchard2011fast}.

In a nutshell, this technique assigns to the edges of a 2-edge-connected graph labels $\phi(e)$, such that two edges $e$ and $f$ define a cut of size 2 in the graph if and only if $\phi(e)=\phi(f)$. We show that using the labels we can understand the structure of cuts of size 2 in a graph, and compute how many cuts are covered by an edge in $O(D)$ rounds. This gives an $O(D\log^3{n})$-round $O(\log{n})$-approximation algorithm for unweighted 3-ECSS.

\section{2-ECSS} \label{sec:2-full}

In this section, we describe our $\widetilde{O}(D + \sqrt{n})$-round $O(\log{n})$-approximation algorithm for 2-ECSS, showing the following.

\twoECSS*

As explained in Section \ref{sec:2-brief}, our algorithm starts by building an MST, $T$, and then augments its connectivity to 2 using a new algorithm for weighted TAP.
We start by describing the general structure of the weighted TAP algorithm. Next, in Section \ref{sec:2-impl-full}, we explain how to implement it efficiently based on a decomposition of the graph into segments.  In Section \ref{sec:dec}, we describe the decomposition in detail. In Section \ref{approx}, we give the approximation ratio analysis, and in Section \ref{sec:iter} we analyze the number of iterations in the algorithm.

Our weighted TAP algorithm follows the framework described in Section \ref{sec:frame}. In the algorithm we maintain a set of edges $A$ that contains all the edges added to the augmentation. 
Recall that for a non-tree edge $e=\{u,v\}$, we denote by $S_e$ the set of tree edges covered by $e$, which are exactly all the tree edges in the unique tree path between $u$ and $v$.
In addition, $C_e$ is the set of tree edges in $S_e$ that are still not covered by edges added to $A$, and the cost-effectiveness of a non-tree edge is $\rho(e) = \frac{|C_e|}{w(e)}$. If $w(e) = 0$, $\rho(e)$ is defined to be $\infty$. However, at the beginning of the algorithm we add to $A$ all the edges with weight $0$, so in the rest of the algorithm we compute cost-effectiveness only for edges where $w(e) \neq 0$.

Our algorithm proceeds in iterations, where in each iteration the following is computed:\\

\noindent\fbox{%
    \parbox{\textwidth}{%
    \vspace{-0.3cm}
\begin{enumerate}
\item { Each non-tree edge $e \not \in A$ computes its rounded cost-effectiveness $\tilde{\rho}(e)$.} \label{cost-ef}
\item { Each non-tree edge $e \not \in A$ with maximum rounded cost-effectiveness is a \textit{candidate}}. \label{max-ce}
\item { Each candidate $e$ chooses a random number $r_e \in \{1,...,n^8\}$}. \label{rand}
\item { Each uncovered tree-edge that is in the set $C_e$ for at least one of the candidates $e$, votes for the first of these candidates, according to the order of the values $r_e$. If there is more than one candidate with the same minimum value, it votes for the one with minimum ID.} \label{first_can}
\item { Each candidate $e$ which receives at least $\frac{|C_e|}{8}$ votes from edges in $C_e$ is added to $A$.}\label{number_votes}
\item { Each tree edge learns if it is covered by edges added to $A$. If all the tree edges are covered, the algorithm terminates, and the output is all the edges added to $A$ during the algorithm.} \label{is_covered}
\end{enumerate}
    \vspace{-0.4cm}
}} \\

We next describe how to implement each iteration in $O(D + \sqrt{n})$ rounds. Later, we show that the algorithm guarantees an $O(\log{n})$-approximation, and that the number of iterations is $O(\log^2{n})$ w.h.p. 
Note that after each edge knows its cost-effectiveness, computing the maximum rounded cost-effectiveness takes $O(D)$ rounds using the BFS tree, and computing the values $r_e$ is a completely local task. Similarly, after each tree edge knows if it is covered, learning if all the edges are covered takes $O(D)$ rounds. We next explain how to implement the rest of the algorithm.

\subsection{Efficient implementation} \label{sec:2-impl-full}

In our algorithm there are basically two types of computation: 
Each tree edge learns information related to all the edges that cover it (which is the first candidate that covers it, is it covered), and each non-tree edge learns information related to the edges it covers in the tree (how  many edges in $S_e$ are not covered, how many edges vote for it). To compute it efficiently, we use a decomposition similiar to the decomposition presented for the FT-MST problem \cite{ghaffari2016near}.

\subsubsection*{Overview of the decomposition} 

In Section \ref{sec:dec}, we show how to decompose the tree into segments satisfying nice properties. We follow the decomposition in \cite{ghaffari2016near} (see Section 4.3) with slight changes. The main difference is that in \cite{ghaffari2016near}, the first step of the decomposition is to choose random edges, and we instead choose edges deterministically using the MST algorithm of Kutten and Peleg \cite{kutten1998fast} which gives a \emph{deterministic} algorithm. The time complexity of our algorithm is $O(D + \sqrt{n}\log^*{n})$ rounds compared to the $O(D+\sqrt{n}\log{n})$-round algorithm in \cite{ghaffari2016near}. Using this decomposition combined with the FT-MST algorithm in \cite{ghaffari2016near} gives a deterministic algorithm for the FT-MST problem in $O(D+\sqrt{n}\log^*{n})$ rounds, which matches the time complexity of computing an MST. 

In the decomposition, the tree is decomposed into $O(\sqrt{n})$ \emph{edge-disjoint} segments with diameter $O(\sqrt{n})$. 
Each segment $S$ has a root $r_S$, which is an ancestor of all the vertices in the segment. The segment contains a main path between the vertex $r_S$ and a descendant of it $d_S$ and additional subtrees attached to this path that are not connected by an edge to other segments in the tree. We call the main path of $S$ the \emph{highway} of $S$, and we call $d_S$ the \emph{unique descendant} of the segment $S$. The vertices $r_S$ and $d_S$ can be a part of other segments (If $r_S \neq r$ it is a unique descendant of another segment, and both $r_S$ and $d_S$ can be roots of additional segments), but other vertices in the segment are not connected by an edge to any other vertex outside the segment. The id of the segment is the pair $(r_S,d_S)$. 
See Figure \ref{decomp} for an illustration.

\setlength{\intextsep}{0pt}
\begin{figure}[h]
\centering
\setlength{\abovecaptionskip}{-2pt}
\setlength{\belowcaptionskip}{6pt}
\includegraphics[scale=0.55]{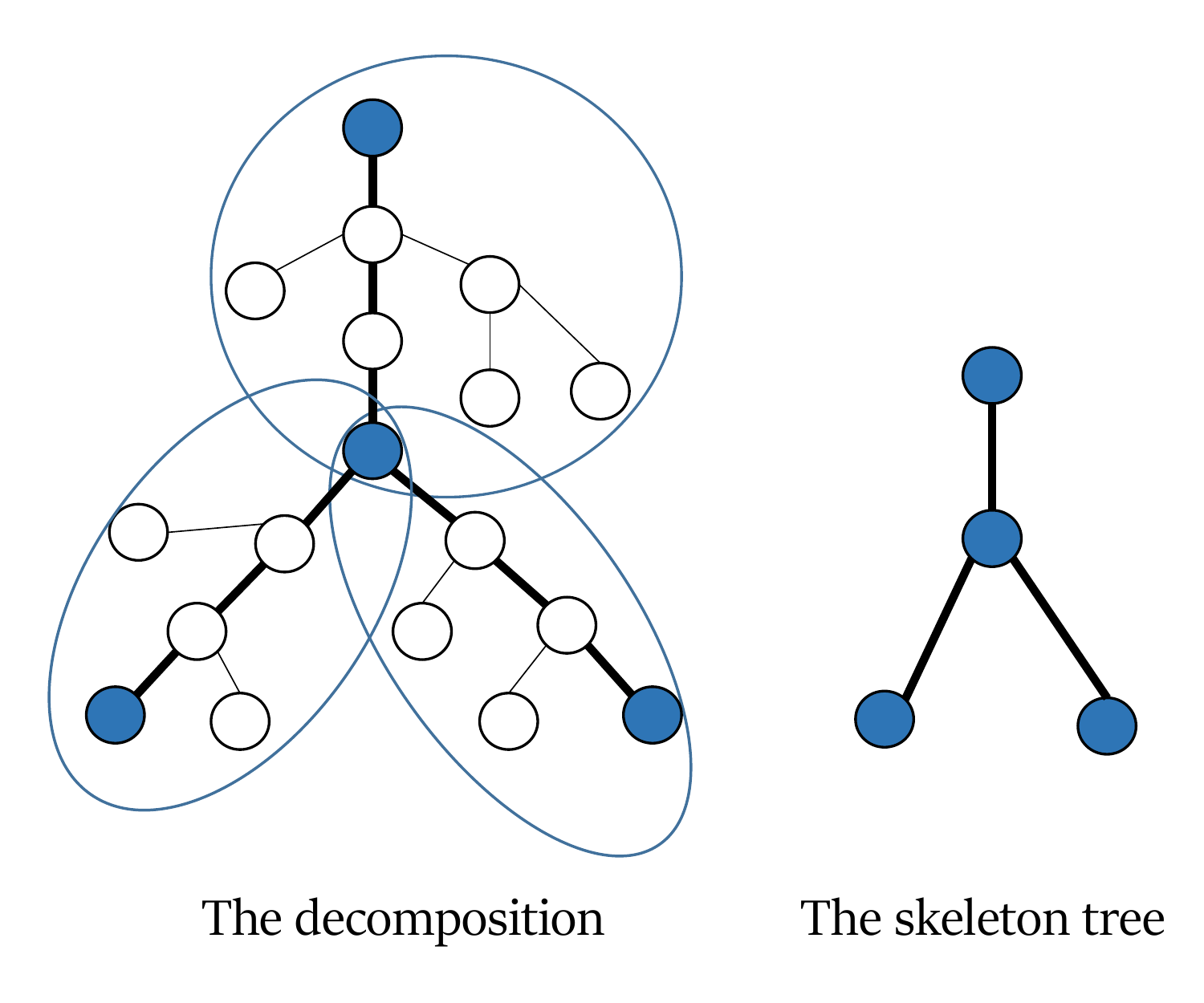}
 \caption{An illustration of the decomposition. The blue vertices are ancestors or unique descendants of segments, and the bold edges are highway edges. Note that the blue vertices can be a part of several segments, but other vertices are not connected by an edge to any vertex outside their segment. The edges of the skeleton tree correspond to highways in the original graph.}
\label{decomp}
\end{figure}

The relatively simple structure of the segments, and in particular the fact that $r_S$ and $d_S$ are the only vertices in $S$ that are connected to other segments, will be very useful in our algorithm. 
The \emph{skeleton tree} $T_S$ is a virtual tree that its vertices are all the vertices that are either $r_S$ or $d_S$ for at least one segment $S$.  The edges in $T_S$ correspond to the highways of the segments, as follows. A vertex $v$ is a \emph{parent} of the vertex $u \neq v$ in the skeleton tree if and only if $v=r_S$ and $u=d_S$ for some segment $S$. 
The time complexity for constructing the segments is $O(D+\sqrt{n}\log^*{n})$ rounds.
Let $P_{u,v}$ be the unique tree path between $u$ and $v$.
In Section \ref{sec:dec}, we show the following.

\begin{restatable}{claim}{info} \label{info1}
In $O(D+\sqrt{n})$ rounds, the vertices learn the following information. All the vertices learn the id of their segment, and the complete structure of the skeleton tree. In addition, each vertex $v$ in the segment $S$ learns all the edges of the paths $P_{v,r_S}$ and $P_{v,d_S}$.
\end{restatable}

\begin{restatable}{claim}{infoTwo} \label{info2}
Assume that each tree edge $t$ and each segment $S$, have some information of $O(\log{n})$ bits, denote them by $m_t$ and $m_S$, respectively. In $O(D+\sqrt{n})$ rounds, the vertices learn the following information. 
Each vertex $v$ in the segment $S$ learns the values $(t,m_t)$ for all the tree edges in the highway of $S$, and in the paths $P_{v,r_S},P_{v,d_S}$. In addition, all the vertices learn all the values $(S,m_S)$. 
\end{restatable}

\subsubsection*{(I). Computing cost-effectiveness} \label{sec:cost-ef}

In Line \ref{cost-ef} of the algorithm, each non-tree edge computes its rounded cost-effectiveness. We next explain how all the non-tree edges compute simultaneously their cost-effectiveness. Rounding the values is a completely local task.

To compute cost-effectiveness, each non-tree edge $e$ should learn how many tree edges in $S_e$ are still not covered.
Before this computation, each tree edge knows if it is covered or not. This holds during the algorithm, since tree edges learn it at the end of the previous iteration (see Line \ref{is_covered}). At the beginning of the algorithm we add to the augmentation all the edges of weight $0$ and each tree edge should learn if it is covered by an edge of weight $0$. To learn it we apply the computation of Line \ref{is_covered} also before the first iteration.

First, each vertex $v$ in the segment $S$ learns the id of $S$, and which edges in the paths $P_{v,r_S},P_{v,d_S}$ and in the highway of $S$ are still not covered. In addition, all the vertices learn the complete structure of the skeleton tree and how many uncovered edges are in the highway of each segment. All the vertices learn it in $O(D+\sqrt{n})$ rounds according to Claims \ref{info1} and \ref{info2}, by having $m_t$ indicate whether the edge $t$ is covered and $m_S$ is the number of uncovered edges in the highway of $S$ (each segment $S$ computes $m_S$ locally in $O(\sqrt{n})$ rounds by scanning the highway of $S$).  

Let $e=\{u,v\}$ be a non-tree edge, and let $LCA(u,v)$ be the lowest common ancestor of $u$ and $v$. The path that $e$ covers in the tree consists of the two paths from $u$ to $LCA(u,v)$ and from $v$ to $LCA(u,v)$. There are three cases:

\emph{Case 1: $u$ and $v$ are in the same segment $S$.} In this case $u$ sends to $v$ all the information about the path $P_{u,r_S}$: all the edges in the path and which of them are covered. Similarly, $v$ sends to $u$ all the information about the path $P_{v,r_S}$. From this information, both $u$ and $v$ can compute their LCA and learn how many uncovered edges are in the path they cover in the tree. The LCA is the vertex where the paths $P_{r_S,u},P_{r_S,v}$ diverge. 
The time complexity is $O(\sqrt{n})$ rounds for sending the paths.

\emph{Case 2: $u$ and $v$ are in different segments, and $LCA(u,v)$ is in another segment.} In this case, the path that $e$ covers consists of the 3 paths $P_{u,r_u},P_{v,r_v},P_{r_v,r_u}$, where $r_v$ and $r_u$ are the ancestors in the segments of $v$ and $u$, respectively. 
Vertex $u$ knows how many uncovered edges are in $P_{u,r_u}$, and $v$ knows how many uncovered edges are in $P_{v,r_v}$. The path $P_{r_u,r_v}$ is composed of highways, as follows. The vertices $r_u$ and $r_v$ are vertices in the skeleton tree and have a path between them in the skeleton tree. This path is the unique path between $r_u$ and $r_v$ in the tree, where we replace each edge of the skeleton tree by the corresponding highway.
Since $v$ and $u$ know the structure of the skeleton tree, as well as how many uncovered tree edges are in each highway, they learn how many uncovered tree edges are in $P_{r_u,r_v}$. By exchanging one message between them, they both learn how many uncovered tree edges are in $P_{u,v}$.

\emph{Case 3: $u$ and $v$ are in different segments, and $LCA(u,v)$ is in the same segment of one of them.} Assume without loss of generality that $LCA(u,v)$ is in the same segment as $u$. Let $d_u$ be the unique descendant in the segment of $u$. 
From the structure of the segments, the only two vertices in this segment that are connected directly to other segments are $r_u$ and $d_u$. If $d_u$ is not an ancestor of $v$ it follows that $LCA(u,v) = r_u$. This case can be solved the same as case 2, since the path $P_{u,v}$ now consists of the 3 paths $P_{u,r_u},P_{v,r_v},P_{r_v,r_u}$. 

Otherwise, $d_u$ is an ancestor of $v$, and the path $P_{u,v}$ consists of the 3 paths $P_{v,r_v},P_{r_v,d_u},P_{d_u,u}$. 
$v$ knows all the information about $P_{v,r_v}$ and $u$ knows all the information about $P_{u,d_u}$. The path $P_{r_v,d_u}$ is composed of highways, since both $r_v$ and $d_u$ are vertices in the skeleton tree. Since $v$ and $u$ know the structure of the skeleton tree and the number of uncovered edges on each highway, they know how many uncovered edges are in $P_{r_v,d_u}$. By exchanging one message between them, $u$ and $v$ learn how many uncovered tree edges are in $P_{u,v}$.

Note that since $u$ and $v$ know the ids of their segments, as well as the structure of the skeleton tree, they can distinguish between the different cases. The whole computation of the cost-effectiveness takes $O(D+\sqrt{n})$ rounds.

\subsubsection*{(II). Learning the optimal edge that covers a tree edge}

Here we explain how all the tree edges learn simultaneously information related to all the non-tree edges that cover them. We need this twice in the algorithm. First, in Line \ref{first_can}, when each tree edge needs to learn which is the first candidate that covers it (all the non-tree edges know if they are candidates, and what is their number after Line \ref{rand}). Second, in Line \ref{is_covered}, when a tree edge needs to learn if it is covered by an edge added to the augmentation (the non-tree edges know if they were added to the augmentation after Line \ref{number_votes}). This can be seen as learning which is the first edge that covers it from those added to the augmentation, if such exists. (When edges are represented as ordered pairs, and we compare them according to the ids of their vertices in lexicographic order). We next explain how each tree edge learns the \emph{best} edge that covers it, where best refers to the first candidate (in Line \ref{first_can}) or to the first from the edges added to the augmentation (in Line \ref{is_covered}).
Similar computations are done in \cite{ghaffari2016near} (see Section 4.2). We include the details for completeness. 

As is done in \cite{ghaffari2016near}, for each tree edge $t$, we classify the non-tree edges that cover it according to 3 types: \emph{short-range}, \emph{mid-range} and \emph{long-range}, as follows. An edge $e=\{u,v\}$ that covers $t$ is a short-range edge for $t$ if both $u$ and $v$ are in the same segment as $t$, it is a mid-range edge if one of $u$ of $v$ is in the same segment as $t$ and it is a long-range edge if neither of them is in the segment of $t$.
In the algorithm, each tree edge $t$ learns which is the best non-tree edge that covers it for each of these 3 types, the optimal of them is the best edge that covers it.

\paragraph*{Learning the optimal short-range edge:} 
Let $t=\{v,p(v)\}$ be a tree edge, where $p(v)$ is the parent of $v$ in the tree. Note that each non-tree edge $e$ that covers $t$ has exactly one vertex at the subtree rooted at $v$, which means that there is a descendant $u$ of $v$ that knows about the edge $e$. If $e$ is a short range edge, $u$ is in the same segment of $t$. The vertex $u$ knows that $e$ covers $t$. This holds since if $e=\{u,w\}$ is a non-tree edge where both $u$ and $w$ are in the same segment $S$, $u$ and $w$ can learn which tree edges of the segment are covered by $e$ by exchanging between them the paths $P_{u,r_S},P_{w,r_S}$ (as explained also in Case 1 of the cost-effectiveness computation). 

The algorithm for computing the optimal short-range edge works as follows. Each vertex $u$ computes the best short-range edge adjacent to it that covers the tree edge $t=\{v,p(v)\}$ for each $v$ where $v$ is an ancestor of $u$, and $t$ is in the same segment as $u$. Each leaf of the segment sends these edges to its parent. Each internal vertex computes for each tree edge $t$ above it in the segment the optimal edge that covers it from the edges it receives from its children and the edges adjacent to it. 
At the end, each tree edge learns about the optimal short-range edge that covers it. Since the diameter of each segment is $O(\sqrt{n})$, and since we can pipeline the computations, the time complexity is $O(\sqrt{n})$ rounds. 

\paragraph*{Learning the optimal long-range edge:}
Note that if a tree edge $t=\{v,p(v)\}$ has a long-range edge that covers it, then $t$ must be on a highway of a segment. As otherwise, all the descendants of $v$ are in the segment of $v$. In addition, the following holds (see Proposition 4.2 in \cite{ghaffari2016near}).

\begin{observation}
All the edges on the highway of a certain segment have the same optimal long-range edge. 
\end{observation} 

This follows since the only vertices in a certain segment $S$ that are connected to other segments are $r_S$ and $d_S$. Therefore, any long-range edge $e$ that covers a tree edge on $S$ has one vertex $u$ that is a descendant of $d_S$, and additional vertex $w \not \in S$ where $LCA(u,w)$ is an ancestor of $r_S$ (that could be $r_S$), which means that $e$ covers all the edges on the highway of $S$.

Hence, in order that all the tree edges would learn about the optimal long-range edge that covers them it is enough to learn the optimal $O(\sqrt{n})$ long-range edges that cover the highways of the segments. Since all the vertices know the complete structure of the skeleton tree, the vertices of each non-tree edge $e=\{u,w\}$ learn which highways of other segments they cover. To learn about the optimal edge that covers a certain highway we can do a convergecast over the BFS tree. Using pipelining, the root $r$ of the BFS tree learns about the optimal edges that cover \emph{all} the highways in $O(D + \sqrt{n})$ rounds. Then, it sends these edges to all the vertices in $O(D + \sqrt{n})$ rounds. 

\paragraph*{Learning the optimal mid-range edge:}
Let $t=\{v,p(v)\}$ be a tree edge covered by a mid-range edge $e$. One of the vertices of $e$ is a descendant of $v$ (it could be $v$ itself), denote it by $u$. There are two cases.

\emph{Case 1: $u$ is in the same segment as $t$.} In this case $u$ knows that the edge $e$ covers $t$, and can pass this information to $v$. In general, all the tree edges can learn about the best mid-range edge that covers them from those that are in case 1 exactly as they learn about the best short-range edge. 

\emph{Case 2: $u$ is in another segment.} Since $e=\{u,w\}$ is a mid-range edge, it follows that $w$ is in the same segment of $t$. In addition, since the segment of $t$ has a unique descendant $d$, it follows that $u$ is a descendant of $d$. The path that $e$ covers is composed of $P_{u,d},P_{d,LCA(w,d)},P_{LCA(w,d),w}$. The tree edges in $P_{u,d}$ are not in the segment of $t$, for the tree edges of $P_{LCA(w,d),w}$, the edge $e$ is in case 1, since $w$ is a descendant of them in the same segment. Hence, $e$ is in case 2 only for the highway edges on the path $P_{d,LCA(w,d)}$.  
Note that $LCA(w,d)$ is a vertex on the highway since all the ancestors of $d$ in its segment are on the highway, and $w$ and $d$ are in the same segment. 

To compute the best mid-range edges in case 2 that cover all the tree edges on a highway of a certain segment, we work as follows. For a vertex on the highway $x$, we denote by $T_x$ the subtree rooted at $x$ of vertices in the segment excluding the highway. This is a subtree that is contained only on the segment, where $x$ is the only highway vertex in this subtree. Vertex $x$ learns which is the best mid-range edge adjacent to a vertex in $T_x$ that covers the highway path $P_{x,d}$. Since each such mid-range edge is adjacent to some vertex in $T_x$, $x$ learns it by convergecast over $T_x$. Note that these subtrees are disjoint for different highway vertices, hence these computations are done in parallel for all the subtrees $T_x$. In order to decide which is the best mid-range edges (from case 2) that cover the tree edges on the highway, we work as follows. For the edge $\{v,p(v)\}$ where $p(v)=r_S$ is the root of the segment, this is the best edge adjacent to a vertex in $T_{r_S}$, and $r_S$ learns it. If $p(v) = x \neq r_S$ this is the optimal edge between the best edge adjacent to a vertex in $T_x$ and the best edge that covers the path $P_{p(x),d}$. Vertex $x$ knows which is the first one, and it receives from its parent the second one. This is done as follows, at the beginning $r_S$ sends the relevant edge to its child in the highway, then its child $v$ can compute the optimal edge that covers $P_{v,d}$ and it sends it to its child, and we continue in the same manner until we compute all the optimal mid-range edges. 

The time complexity is proportional to the diameter of the segment, and we do these computations in different segments in parallel which results in a time complexity of $O(\sqrt{n})$ rounds.

Finally, the best edge that covers a tree edge is the best between the best short-range, long-range, and mid-range edges computed. The whole time complexity for learning the best edge is $O(D+\sqrt{n})$ rounds.

\subsubsection*{(III). Computing the number of votes}

In Line \ref{number_votes} of the algorithm, each candidate computes the number of votes it receives from edges in $C_e$. Note that after Line \ref{first_can}, each tree edge $t$ knows which is the first candidate that covers it, if such exists, denote it by $m_t$.
Computing the number of votes is similar to computing the cost-effectiveness. When computing cost-effectiveness, each non-tree edge computed how many tree edges in $S_e$ are still not covered. Now each of the candidates computes how many of these edges vote for it. As explained in the cost-effectiveness computation, the non-tree edge $e=\{v,u\}$ covers paths in the segments of $v$ and of $u$ as well as highway edges in different segments. If a highway edge $t$ is covered by $e$, and $t$ is in a different segment than $v$ and $u$, then $e$ is a long-range edge for $t$.
The edge $t$ votes for $e$ only if $e$ is the best long-range edge that covers the highway of the segment of $t$, and if the long-range edge is the best edge that covers $t$.
Hence, in order to compute how many highway edges of different segments vote for $e$ it is enough that all the vertices learn which is the best long-range edge that covers the highway of each segment (an information all the vertices learn when computing the best long-range edges), and for how many edges of the highway of each segment the long-range edge is the best edge that covers them, denote this number by $m_S$ for a segment $S$. On each segment, the root of the segment learns $m_S$ by a simple scan of the highway (since all the tree edges know which is the best edge that covers them). 

Now all the vertices learn all the values $(S,m_S)$. In addition, each vertex $v$ in the segment $S$ learns the values $(t,m_t)$ for all the tree edges in the highway of $S$, and in the paths $P_{v,r_S},P_{v,d_S}$. All the vertices learn this information in $O(D + \sqrt{n})$ rounds using Claim \ref{info2}. 
Now all the non-tree edges can learn how many tree edges vote for them following the same analysis for computing cost-effectiveness. The overall time complexity is $O(D + \sqrt{n})$ rounds. 

To conclude, all the computations in an iteration take $O(D + \sqrt{n})$ rounds, which shows the following.

\begin{lemma} \label{time-iteration}
Each iteration takes $O(D + \sqrt{n})$ rounds.
\end{lemma}

\subsection{The decomposition} \label{sec:dec}

In this section, we explain how to decompose the tree into segments with the desired properties. The decomposition follows the decomposition in \cite{ghaffari2016near} (see Section 4.3) with slight changes. The main difference is that in \cite{ghaffari2016near}, the first step of the decomposition is to choose random edges, and we instead choose edges deterministically using the MST algorithm of Kutten and Peleg \cite{kutten1998fast} which gives a \emph{deterministic} algorithm. The time complexity of our algorithm is $O(D + \sqrt{n}\log^*{n})$ rounds compared to the $O(D+\sqrt{n}\log{n})$-round algorithm in \cite{ghaffari2016near}. Using this decomposition combined with the FT-MST algorithm in \cite{ghaffari2016near} gives a deterministic algorithm for the FT-MST problem in $O(D+\sqrt{n}\log^*{n})$ rounds, which matches the time complexity of computing an MST. 

\subsubsection*{(I). Preliminary step: defining fragments and virtual tree}
Before defining the segments and skeleton tree, we start by decomposing the tree into \emph{fragments}, that do not necessarily satisfy the desired properties. We later decompose the fragments further, to define the segments and skeleton tree.  
At the beginning of our algorithm, we computed an MST $T$ using the algorithm of Kutten and Peleg \cite{kutten1998fast}, which takes $O(D + \sqrt{n}\log^*{n})$ rounds. This algorithm decomposes the tree into $O(\sqrt{n})$ fragments of diameter $O(\sqrt{n})$, where each vertex knows the id of its fragment. We can assume that $T$ is a rooted tree with root $r$ and each vertex knows its parent in $T$. This can be achieved as follows. 
We say that a tree edge is a \emph{global} edge if it connects two different fragments. Since there are $O(\sqrt{n})$ fragments, there are $O(\sqrt{n})$ global edges and all the vertices learn them in $O(D+\sqrt{n})$ rounds by communicating over the BFS tree. If we contract each fragment to one vertex, the global edges define a temporary virtual tree between the fragments. The virtual tree is a rooted tree where the root is the fragment of $r$, and all the vertices know the complete structure of the virtual tree since they know all the global edges. In each fragment, we define the root of the fragment to be the vertex closest to $r$ (which can be deduced from the structure of the virtual tree). All the roots of the fragments know their parents in $T$ from the virtual tree. By communication on each fragment separately, each internal vertex in the fragment learns its parent in the fragment, which is its parent in $T$. The time complexity is proportional to the diameter of the fragment which is $O(\sqrt{n})$.

\subsubsection*{(II). Marking vertices}
To define segments, we start by marking vertices. Our algorithm follows the decomposition in \cite{ghaffari2016near} (see Section 4.3, steps 2 and 3) where the global edges play the role of the sampled edges $R$ in \cite{ghaffari2016near}. We mark all the vertices that are endpoints of some global edge, and we mark the root $r$. Next, we mark also all the LCAs of marked vertices in $O(\sqrt{n})$ rounds, as follows. We scan each fragment from the leaves to the root. A leaf $v$ sends to its parent the id of $v$ if $v$ is marked, or $\emptyset$ otherwise. An internal vertex $v$ that receives from its children at most one id, passes this id to its parent (or $\emptyset$ if it did not receive any id), and if it receives more than one id, it marks itself, and passes one of these ids to its parent. 
We next show the following Lemma (equivalent to Lemma 4.4 in \cite{ghaffari2016near}). The proof is similar to the proof in \cite{ghaffari2016near}, and is included here for completeness.

\begin{lemma}
The set of marked vertices satisfies the following properties:
\begin{enumerate}
\item The root $r$ is marked, and for each vertex $v \neq r$, there is a marked ancestor of $v$ at distance at most $O(\sqrt{n})$ from $v$. \label{roots}
\item For each two marked vertices $u$ and $v$, their LCA is also a marked vertex. \label{LCAS}
\item The number of marked vertices is $O(\sqrt{n})$. \label{marked}
\end{enumerate}
\end{lemma}

\begin{proof}
The root is marked by definition, and for each vertex $v \neq r$, the root of $v$'s fragment is an ancestor of $v$ at distance at most $O(\sqrt{n})$ from $v$. This proves (\ref{roots}).

Let $u$ and $v$ be two marked vertices, if $u$ and $v$ are in the same fragment, their LCA is also in this fragment and it received at least two different ids in the algorithm, hence it is a marked vertex. If they are in different fragments, let $F$ be the fragment that includes their LCA. If one of $u$ and $v$ is in $F$, assume without loss of generality that $u \in F$, and let $v_F$ be the first vertex in $F$ in the tree path $P_{v,LCA(u,v)}$. Note that $v_F$ is a marked vertex since it is adjacent to a global edge. Now $LCA(u,v)=LCA(u,v_F)$ is a marked vertex since it is an LCA of two marked vertices in the fragment $F$. If neither of $u$ and $v$ is in $F$, let $u_F$ and $v_F$ be the first vertices in $F$ in the paths $P_{u,LCA(u,v)},P_{v,LCA(u,v)}$. Both $u_F$ and $v_F$ are marked vertices in $F$ since they are adjacent to global edges. Hence, $LCA(u,v)=LCA(u_F,v_F)$ is also a marked vertex.
This completes the proof of (\ref{LCAS}).

We now prove (\ref{marked}). Since there are $O(\sqrt{n})$ global edges, there are $O(\sqrt{n})$ marked vertices at the beginning, we call them the \emph{original} marked vertices. Let $v$ be a non-original marked vertex. Then, $v$ receives more than one id from its children, and passes to its parent one of these ids. Note that all the ids sent in the algorithm are ids of original marked vertices. We give $v$ a label $l_v$ that includes an id of an original marked vertex it receives and does not send to its parent. Note that the id $l_v$ does not reach any other marked vertex above $v$ and vertices below $v$ sent this id forward. Hence, different marked vertices receive different labels. Since all the labels are ids of original marked vertices, there are at most $O(\sqrt{n})$ non-original marked vertices. This completes the proof of (\ref{marked}). 
\end{proof}

\subsubsection*{(III). Defining the segments and the skeleton tree}
We next define the segments according to the marked vertices. For each marked vertex $d_S \neq r$, the tree path between $d_S$ to its closest marked ancestor $r_S \neq d_S$ defines a highway of a segment $S$. Note that all the internal vertices in the highway of $S$ are not marked by definition. They also do not have any marked descendants. Otherwise, there is an internal vertex $v$ in the highway that is the LCA of $d_S$ and another marked vertex $d$, which contradicts the fact that $v$ is not marked. The segment $S$ includes the highway between $r_S$ and $d_S$ (that are the ancestor and unique descendant of the segment), as well as all the descendants of internal vertices in the highway. 
Note that $d_S$ and $r_S$ can be a part of other highways.
Let $v$ be a marked vertex, then some of its children may be included in highways and are already included in a segment. If $v$ has children that are not included in any highway, it means that they do not have any marked descendant. All these children and the subtrees rooted at them are added to a segment with root $v$, as follows. If $v$ is already a root of at least one segment $S$, all these vertices are added to $S$. Otherwise, we define a new segment with an empty highway that has all these vertices, the id of the segment is $(v,v)$. 
Since there are $O(\sqrt{n})$ marked vertices, this decomposes $T$ into $O(\sqrt{n})$ edge-disjoint segments. The diameter of a segment is $O(\sqrt{n})$ because each vertex has a marked ancestor at distance at most $O(\sqrt{n})$, and $r_S$ is an ancestor of all the vertices in $S$ by definition.

The \emph{skeleton tree} $T_S$ is a virtual tree that its vertices are all the marked vertices. The edges in $T_S$ correspond to the highways of the segments, as follows. A vertex $v$ is a parent of the vertex $u \neq v$ in the skeleton tree if and only if $v=r_S$ and $u=d_S$ for some segment $S$. 
We next explain how vertices learn information about their segments and the skeleton tree.   

\subsubsection*{(IV). Learning information about the segments and skeleton tree}

We next prove Claims \ref{info1} and \ref{info2}.

\info*

\begin{proof}
At the beginning, each root of a segment learns the id of the segment, as follows. Each marked vertex $v \neq r$ sends to its parent its id, and each non-marked leaf sends to its parent $\emptyset$. Internal vertices that receive one id from their children send it to their parent, otherwise they send $\emptyset$. We continue until the messages reach a marked vertex. Note that an internal non-marked vertex receives at most one id from its children, otherwise it is the LCA of two marked vertices, and it is marked. At the end, each root $r_S$ of a segment $S$ knows $d_S$. The vertex $r_S$ broadcasts the id $(r_S,d_S)$ to all the vertices of the segment. The time complexity is $O(\sqrt{n})$ rounds, since this is the diameter of the segments. All the vertices learn all the ids $(r_S,d_S)$ in $O(D+\sqrt{n})$ rounds by communication over the BFS tree. Hence, all the vertices know the id of their segment and the complete structure of the skeleton tree.

Now each vertex $v$ in the segment $S$, learns all the edges in $P_{v,r_S}$, as follows. At the beginning, each parent sends to its children its id. Then, it sends the id it gets in the previous round. We continue in the same manner for $O(\sqrt{n})$ rounds, which allows each vertex to learn all its ancestors in its segment. If we perform the exact same computation in the reverse order (where vertices send ids to their parent) only over the highway edges, the root $r_S$ learns all the edges of the highway of $S$. Then, it broadcasts this path of length $O(\sqrt{n})$ to all the vertices in the segment. This gives to all the vertices in $S$ full information about the highway of $S$. Note that $P_{v,d_S}$ is composed of $P_{v,LCA(v,d_S)}$ and $P_{d_S,LCA(v,d_S)}$. The first is contained in $P_{v,r_S}$, the second is contained in the highway of $S$. Since $v$ knows all the information about $P_{v,r_S}$ and the highway of its segment, it learns $LCA(v,d_S)$ and learns the path $P_{v,d_S}$. This completes the proof of Claim \ref{info1}.
\end{proof}

Assume now that each tree edge $t$ has some information $m_t$ of $O(\log{n})$ bits and each segment has some information $m_S$ of $O(\log{n})$ bits. Then each vertex $v \in S$ learns the values $(t,m_t)$ for all the tree edges in the highway of $S$, and in the paths $P_{v,r_S},P_{v,d_S}$ exactly in the same manner. To learn all the values $(S,m_S)$ we communicate over the BFS tree. 
The whole time complexity for learning the information is $O(D+\sqrt{n})$ rounds. This shows the following.

\infoTwo*

\subsection{Approximation ratio analysis} \label{approx} 

In this section, we show that our algorithm for weighted TAP guarantees an $O(\log{n})$-approximation. 
Some elements of our analysis have similar analogues in the classic analysis of the greedy set cover algorithm \cite{johnson1974approximation, chvatal1979greedy, lovasz1975ratio}. We also used similar ideas in a recent algorithm for the minimum 2-spanner problem \cite{spanner}. 

Let $A$ be the set of edges added to the augmentation by the algorithm. When the algorithm ends, all the tree edges are covered, hence $T \cup A$ is 2-edge-connected. We show that $w(A) \leq O(\log{n})w(A^*)$, where $A^*$ is an optimal augmentation. 

To show the approximation ratio, we assign each edge $t \in T$ a value $cost(t)$ such that the sum of the costs of all edges is closely related both to $w(A)$ and $w(A^*)$, by satisfying $$w(A) \leq 8 \sum_{t \in T} cost(t) \leq O(\log{n}) w(A^*),$$
which implies our claimed approximation ratio.

For a tree edge $t$, let $i$ be the iteration in which $t$ is first covered in the algorithm.
The edge $t$ may be covered by a candidate edge $e$ that it votes for and is added to $A$ at iteration $i$.  In this case, we set $cost(t) = \frac{1}{\rho(e)}$, where $\rho(e)$ is the cost-effectiveness of the edge $e$ at iteration $i$. Another option is that $t$ is covered by other edges added to $A$ at iteration $i$, or at the beginning of the algorithm since it is covered by an edge of weight 0.
In these cases, we set $cost(t) = 0$.
We first show the left inequality above.

\begin{lemma} \label{cost}
$w(A) \leq 8 \cdot \sum_{t \in T} cost(t)$.
\end{lemma}

\begin{proof}
Let $e$ be an edge with non-zero weight that is added to $A$ at iteration $i$, and let $\rho(e)$ be the cost-effectiveness of $e$ at iteration $i$. Recall that we add $e$ to $A$ since it gets at least $\frac{|C_e|}{8}$ votes from the tree edges it covers. Denote by $Votes(e)$ the set of tree edges that vote for $e$ at iteration $i$. For each $t \in Votes(e)$, we defined $cost(t) = \frac{1}{\rho(e)}$, which gives, $$\sum_{t \in Votes(e)} cost(t) \geq \frac{1}{\rho(e)} \cdot \frac{|C_e|}{8} =  \frac{w(e)}{|C_e|} \cdot \frac{|C_e|}{8} = \frac{w(e)}{8}.$$
Hence, for each $e \in A$, $w(e) \leq 8 \cdot \sum_{t \in Votes(e)} cost(t)$ (for an edge with weight 0, this holds trivially).
For each tree edge $t$, there is at most one edge $e \in A$ such that $t \in Votes(e)$, since $t$ votes for at most one edge at the iteration in which it is covered. Hence, we get $$w(A) = \sum_{e \in A} w(e) \leq 8 \cdot \sum_{e \in A} \sum_{t \in Votes(e)} cost(t) \leq 8 \cdot \sum_{t \in T} cost(t).$$
This completes the proof of Lemma \ref{cost}.
\end{proof}

\begin{lemma} \label{weighted}
$\sum_{t \in T} cost(t) \leq O(\log{n})w(A^*)$.
\end{lemma}


\begin{proof}
Consider an edge $e \in A^*$ and let $(t_1,...,t_{\ell})$ be the sequence of tree edges covered by $e$ according to the order in which they are covered in the algorithm. Assume first that $w(e) \neq 0$.
The cost-effectiveness of $e$ at the beginning of the iteration in which $t_1$ is covered is $\frac{\ell}{w(e)}$. All the candidates that cover $t_1$ have maximum rounded cost-effectiveness. In particular, the cost-effectiveness of the edge that covers $t_1$ is at least $\frac{\ell}{2w(e)}$. Hence, $cost(t_1) \leq \frac{2w(e)}{\ell}$.
Similarly, the cost-effectiveness of $e$ at the beginning of the iteration in which $t_j$ is covered is at least $\frac{\ell - j +1}{w(e)}$, which gives $cost(t_j) \leq \frac{2w(e)}{\ell - j + 1}$.
This gives, $$\sum_{j=1}^{\ell} cost(t_j) \leq 2w(e) \cdot \sum_{j=1}^{\ell} \frac{1}{\ell - j + 1} = O(\log{\ell})w(e) = O(\log{n})w(e).$$
The last equality is because the number of tree edges $\ell$ covered by $e$ is at most $n$. 

For an edge $e \in A^*$ such that $w(e)=0$, note that $cost(t)=0$ for all the edges covered by $e$, since they are all covered at the beginning of the algorithm without voting for any candidate. 
Hence, we get in this case $\sum_{j=1}^{\ell} cost(t_j) = 0 = O(\log{n})w(e).$

Recall that $S_e$ is the set of tree edges covered by the edge $e$. Since $A^*$ is an augmentation, every edge $t \in T$ is covered by at least one edge $e \in A^*$. Summing over all the edges in $A^*$ we get,
$$\sum_{t \in T} cost(t) \leq \sum_{e \in A^*} \sum_{t \in S_e} cost(t) \leq O(\log{n})\sum_{e \in A^*} w(e) = O(\log{n})w(A^*).$$
This completes the proof.
\end{proof}

In conclusion, we get $w(A) \leq 8 \sum_{t \in T} cost(t) \leq O(\log{n}) w(A^*)$, which completes the proof of the $O(\log{n})$-approximation ratio for weighted TAP, giving the following lemma.

\begin{lemma} \label{approx-tap}
The approximation ratio of the weighted TAP algorithm is $O(\log{n})$.
\end{lemma}

\subsection{Analyzing the number of iterations} \label{sec:iter}

We next show that the number of iterations of our weighted TAP algorithm is $O(\log^2{n})$ w.h.p. 
The analysis is based on a potential function argument which is described in \cite{jia2002efficient, rajagopalan1998primal} for the set cover and minimum dominating set problems. 

Let $\tilde{\rho}$ be the maximum rounded cost-effectiveness at the beginning of iteration $i$. We define the potential function $\phi = \sum_{e:\tilde{\rho}(e)=\tilde{\rho}}|C_e|$, where $C_e$ are the edges in $S_e$ that are still not covered by $A$ at the beginning of iteration $i$. We say that an iteration is \emph{legal} if the random numbers $r_e$ chosen by the candidates in this iteration are different.
The following lemma shows that if the value of $\tilde{\rho}$ does not change between iterations, and the iterations are legal, the potential function $\phi$ decreases by a multiplicative factor between iterations in expectation. The proof follows the lines of the proofs in \cite{jia2002efficient, rajagopalan1998primal}, and is included here for completeness.

\begin{lemma} \label{dec}
If $\phi$ and $\phi'$ are the potentials at the beginning and end of a legal iteration, then $E[\phi'] \leq c \cdot \phi$ for some positive constant $c<1$.
\end{lemma}

In order to prove Lemma \ref{dec} we need the following definitions. Let $s(t)$ be the number of candidates that cover the tree edge $t$. For a candidate $e$, we sort the tree edges in $C_e$ according to $s(t)$ in non-increasing order. Let $T(e)$ and $B(e)$ be the sets of the first $\lceil |C_e|/2 \rceil$ edges, and the last $\lceil |C_e|/2 \rceil$ edges in the sorted order, respectively. Indeed, if $|C_e|$ is odd, the sets $T(e)$ and $B(e)$ share an edge.
For a pair $(e,t)$, where $e$ is a candidate that covers $t$, we say that $(e,t)$ is \emph{good} if $t \in T(e)$.
We next show that if $t \in T(e)$ chooses $e$ in a legal iteration, then the edge $e$ is added to $A$ with constant probability.

\begin{claim} \label{greater}
Let $i$ be a legal iteration. If $t,t'$ are both covered by $e$ in iteration $i$, and $s(t) \geq s(t')$, then $Pr[t'\ chooses\ e| t\ chooses\ e] \geq \frac{1}{2}$.
\end{claim}

\begin{proof}
Let $N_t,N_{t'},N_b$ be the number of candidates that cover $t$ but not $t'$, $t'$ but not $t$, and both $t$ and $t'$, respectively.
$$Pr[t'\ chooses\ e| t\ chooses\ e] = \frac{Pr[t\ and\ t'\ choose\ e]}{Pr[t\ chooses\ e]} = \frac{\frac{1}{N_t + N_{t'} + N_b}}{\frac{1}{N_t + N_b}} = \frac{|N_t|+|N_b|}{|N_t|+|N_{t'}|+|N_b|}.$$
It holds that $N_t \geq N_{t'}$ since $s(t) \geq s(t')$. This gives,
$$Pr[t'\ chooses\ e| t\ chooses\ e] = \frac{|N_t|+|N_b|}{|N_t|+|N_{t'}|+|N_b|} \geq \frac{|N_t|+|N_b|}{2|N_t|+|N_b|} \geq \frac{1}{2}.$$
\end{proof}

\begin{claim}
If $(e,t)$ is a good pair in a legal iteration $i$, then $Pr[e\ is\ chosen|t\ chooses\ e] \geq \frac{1}{3}$.
\end{claim}

\begin{proof}
Assume that $t$ chooses $e$. Denote by $X$ the number of edges in $B(e)$ that choose $e$, and let $X'=|B(e)|-X$. Note that $t \in T(e)$ since $(e,t)$ is good, therefore $s(t) \geq s(t')$ for any edge $t' \in B(e)$.
By Claim \ref{greater}, any edge $t' \in B(e)$ chooses $e$ with probability at least $\frac{1}{2}$. Hence, $E[X] \geq \frac{|B(e)|}{2}$. Equivalently, $E[X'] \leq \frac{|B(e)|}{2}$. Using Markov's inequality we get
$$Pr[X < \frac{|B(e)|}{4}]=Pr[X' > \frac{3}{4}|B(e)|] \leq Pr[X' \geq \frac{3}{2}E[X']] \leq \frac{2}{3}.$$
Hence, we get $Pr[X \geq \frac{|B(e)|}{4}] \geq \frac{1}{3}$. Since $|B(e)| \geq \frac{|C_e|}{2}$, it holds that $X \geq \frac{|C_e|}{8}$ with probability at least $\frac{1}{3}$. In this case, at least $\frac{|C_e|}{8}$ edges choose $e$, and it is added to $A$. This completes the proof.
\end{proof}

We can now bound the value of the potential function, by proving Lemma \ref{dec}.

\begin{proof} [Proof of Lemma \ref{dec}.]
Let $\phi$ and $\phi'$ be the values of the potential function at the beginning and end of a legal iteration $i$.
It holds that $\phi = \sum_{e:\tilde{\rho}(e)=\tilde{\rho}}|C_e|= \sum_{(e,t)} 1 = \sum_{t} s(t)$, where we sum over all the tree edges that are in $C_e$ for at least one candidate $e$, and over all the pairs $(e,t)$ where $e$ is a candidate and $t \in C_e$. Note that the rounded cost-effectiveness of all the candidates is $\tilde{\rho}$. If the edge $t$ chooses $e$, and $e$ is added to $A$, $\phi$ decreases by $s(t)$. We ascribe this decrease to the pair $(e,t)$. Since $t$ chooses only one candidate, any decrease in $\phi$ is ascribed only to one pair. Hence, we get
\begin{flalign*}
E[\phi - \phi'] &\geq \sum_{(e,t)} Pr[t\ chooses\ e,e\ is\ chosen]\cdot s(t) \\
&\geq \sum_{(e,t)\ is\ good} Pr[t\ chooses\ e]\cdot Pr[e\ is\ chosen|t\ chooses\ e]\cdot s(t) \\
&\geq \sum_{(e,t)\ is\ good} \frac{1}{s(t)} \cdot \frac{1}{3} \cdot s(t) = \frac{1}{3} \sum_{(e,t)\ is\ good} 1.
\end{flalign*}
Since at least half of the pairs are good, we get $E[\phi - \phi'] \geq \frac{1}{6} \phi$, or equivalently $E[\phi'] \leq \frac{5}{6} \phi$, which completes the proof.
\end{proof}

In conclusion, we get the following.

\begin{lemma} \label{time}
The time complexity of the weighted TAP algorithm is $O(\log^2{n})$ rounds w.h.p.
\end{lemma}

\begin{proof}
Recall that $\rho(e) = \frac{|C_e|}{w(e)}$.
While $T \cup A$ is not 2-edge-connected, the maximum cost-effectiveness of an edge $e \not \in A$ is between $\frac{1}{w_{max}}$ and $\frac{n}{w_{min}}$ where $w_{min},w_{max}$ are the minimum and maximum positive weights of an edge. 
Since the cost-effectiveness values are rounded to powers of 2, and since the weights are polynomial,
$\tilde{\rho}$ may obtain at most $O(\log{(n \cdot \frac{w_{max}}{w_{min}})})=O(\log{n})$ values. In addition, by Lemma \ref{dec}, if $\tilde{\rho}$ has the same value at iterations $j$ and $j+1$, and $j$ is a legal iteration, then the value of $\phi$ decreases between these iterations by a factor of at least $1/c$ in expectation. Since the random numbers $r_e$ are chosen from $\{1,...,n^8\}$, they are different w.h.p, giving that if $\tilde{\rho}$ has the same value in any two consecutive iterations then the value of $\phi$ decreases between these iterations by a constant factor in expectation.
Since $\phi \leq n^3$, after $O(\log{(n^3)})=O(\log{n})$ iterations in expectation, the value of $\tilde{\rho}$ must decrease. This shows that the time complexity is $O(\log^2{n})$ rounds in expectation. A Chernoff bound then gives that this also holds w.h.p.
\end{proof}

\textbf{Remark:} Our algorithm can work also for arbitrary weights, but then the number of iterations would be $O(\log{n}\log{(n \cdot \frac{w_{max}}{w_{min}})})$, according to the proof of Lemma \ref{time}. Also, if the weights are arbitrarily large we can no longer send a weight in a round, hence the time complexity depends on the number of rounds needed to send a weight.

By Lemma \ref{time}, the number of iterations in the algorithm is $O(\log^2{n})$ w.h.p. Since each iteration takes $O(D + \sqrt{n})$ rounds by Lemma \ref{time-iteration}, we get a time complexity of $O((D + \sqrt{n})\log^2{n})$ rounds w.h.p. The approximation ratio of the algorithm is $O(\log{n})$ by Lemma \ref{approx-tap}, which gives the following.

\begin{theorem}
There is a distributed algorithm for weighted TAP in the \congest model that guarantees an approximation ratio of $O(\log{n})$, and takes $O((D + \sqrt{n})\log^2{n})$ rounds w.h.p.
\end{theorem}

Since our algorithm for 2-ECSS starts by building an MST in $O(D+\sqrt{n}\log^*{n})$ rounds, and then augments it using our weighted TAP algorithm, Claim \ref{Aug_k} shows the following.

\twoECSS* 

\section{$k$-ECSS} \label{k-full}

In this section, we present our algorithm for $k$-ECSS, proving the following.

\kECSS*

As explained in Section \ref{sec:frame}, to solve $k'$-ECSS we present an algorithm for $Aug_k$ for any $k \leq k'$. The input for $Aug_k$ is a $k$-edge-connected graph $G$, and a $(k-1)$-edge-connected spanning subgraph $H$, and the goal is to augment $H$ to be $k$-edge-connected.
We start by describing the general structure of our algorithm for $Aug_k$. Implementation details and time analysis appear in Section \ref{sec:k-time}. The approximation ratio analysis appears in Section \ref{k-approx}.

Throughout the algorithm we maintain a set $A$ of all the edges added to the augmentation. Initially, $A = \emptyset$. We assume that all the vertices know all the edges in $H$ and in $A$ during the algorithm, we later explain how maintaining this knowledge affects the time complexity. 
Our algorithm follows the framework described in Section \ref{sec:frame}. Recall that for an edge $e \not \in H$, we denote by $S_e$ the set of cuts of size $k-1$ of $H$ that $e$ covers, and we denote by $C_e$ all the cuts in $S_e$ that are still not covered by edges added to $A$.
The cost-effectiveness of an edge $e \not \in H$ is $\rho(e) = \frac{|C_e|}{w(e)}$.

The algorithm proceeds in iterations, where in the iteration $i$ the following is computed:\\

\noindent\fbox{%
    \parbox{\textwidth}{%
    \vspace{-0.3cm}
\begin{enumerate}
\item { Each edge $e \not \in H \cup A$ computes its rounded cost-effectiveness $\tilde{\rho}(e)$.} \label{k-cost-ef}
\item { Each edge $e \not \in H \cup A$ with maximum rounded cost-effectiveness is a \textit{candidate}}.  \label{k-max}
\item { Each candidate $e$ becomes an \emph{active} candidate with probability $p_i$}. \label{prob}
\item { The vertices compute an MST of the graph $G$, with the following weights. All the edges already in $A$ have weight $0$, all the active candidates have weight $1$, and all other edges have weight $2$. We add to $A$ all the active candidates that are in the MST computed.} \label{forest} 
\item { If all the cuts of size $k-1$ in $H$ are covered by $A$, the algorithm terminates, and the output is all the edges of $A$.} \label{k-term} 
\end{enumerate}
    \vspace{-0.4cm}
}} \\

We next show that in Line \ref{forest} we add to $A$ a maximal set of active candidates without creating cycles. This guarantees that $|A| \leq n-1$.  
Then, we describe how to choose the probability $p_i$ in Line \ref{prob}. We later explain how to implement the algorithm in $O(D \log^3{n} + n)$ rounds, and prove the approximation ratio.

\begin{claim} \label{no-cycle}
There are no cycles in $A$.
\end{claim}

\begin{proof} 
At the beginning, this clearly holds. Assume it holds at the beginning of iteration $i$. Let $T_i$ be the MST computed in iteration $i$, and let $A_0$ be the edges of $A$ at the beginning of iteration $i$. 
Recall that the weight given to all the edges of $A_0$ in Line \ref{forest} is $0$. This shows that $T_i$ has all the edges of $A_0$, as otherwise there is an edge $e \in A_0 \setminus T_i$. However, adding $e$ to $T_i$ creates a cycle that has an edge $e' \not \in A_0$ (because there are no cycles in $A_0$), but then $(T_i \setminus \{e'\}) \cup \{e\}$ is a spanning tree with a smaller weight than $T_i$, which gives a contradiction. Hence, all the edges we added to $A$ in iteration $i$ are in the tree $T_i$, which shows that there are no cycles in $A$ at the end of the iteration. This completes the proof. 
\end{proof}

Let $\widetilde{A}_i$ be the set of edges in $A$ at the end of iteration $i$.

\begin{claim} \label{cycle}
If $e$ is an active candidate in iteration $i$, and $e \not \in \widetilde{A}_i$, then adding $e$ to $\widetilde{A}_i$ closes a cycle.
\end{claim}

\begin{proof}
Let $T_i$ be the MST computed in iteration $i$. If $e \not \in \widetilde{A}_i$, then $e \not \in T_i$ by definition. Then, adding $e$ to $T_i$ creates a cycle. If all the edges in the cycle except $e$ are in $\widetilde{A}_i$, we are done. Note that all the edges of weight $0$ or $1$ in $T_i$ are in $\widetilde{A}_i$ according to Line \ref{forest} in the algorithm. Hence, an edge $e' \neq e$ in the cycle is not in $\widetilde{A}_i$ only if its weight is 2, but then $(T_i \setminus \{e'\}) \cup \{e\}$ is a spanning tree with a smaller weight than $T_i$, which contradicts the definition of $T_i$. This completes the proof.
\end{proof}

We next show that all the cuts that can be covered by active candidates at iteration $i$, are covered by the end of iteration $i$.

\begin{claim} \label{cuts}
If $e$ is an active candidate in iteration $i$, then all the cuts in $S_e$ are covered by the end of iteration $i$.
\end{claim}

\begin{proof}
Let $e$ be an active candidate in iteration $i$. If $e$ is added to $\widetilde{A}_i$, the claim clearly holds. Otherwise, adding $e$ to $\widetilde{A}_i$ closes a cycle by Claim \ref{cycle}. For each cut in $S_e$, there is also an edge in this cycle that covers the cut, which means that all the cuts in $S_e$ are covered by $\widetilde{A}_i$.
\end{proof}

\subsubsection*{Choosing $p_i$}

We next describe how to choose the probability $p_i$. 
We group the iterations according to the value of the maximum rounded cost-effectiveness. For each value $\tilde{\rho}$ of maximum rounded cost-effectiveness, we divide the iterations into phases of $O(\log{n})$ iterations, as follows. 
At the first iteration, we define $p_i = \frac{1}{2^{\lceil \log{m} \rceil}}$, where $m$ is the number of edges in $G$, and every $M \log{n}$ iterations we increase $p_i$ by a factor of 2. The exact value of the constant $M$ will be determined during the analysis.
We continue until the maximum rounded cost-effectiveness decreases or until the first iteration where $p_i = 1$. Note that if $p_i=1$ then all the candidates are active candidates, hence by Claim \ref{cuts} the maximum rounded cost-effectiveness decreases by the end of this iteration. 
Every $M \log{n}$ consecutive iterations with the same value $p_i$ are a \emph{phase}. The algorithm takes $O(\log^3{n})$ iterations: each phase takes $O(\log{n})$ iterations, and we increase $p_i$ at most $O(\log{n})$ times for each value of rounded cost-effectiveness. In addition, there are $O(\log{n})$ possible values for the rounded cost-effectiveness because the weights are polynomial and there are at most $n \choose 2$ minimum cuts of size $k-1$ in $H$\footnote{This follows from the minimum cut algorithm of Karger \cite{DBLP:conf/soda/Karger93}, another proof is in \cite{dinitz1976structure}.}.
All the candidates can compute $p_i$ since they know the value of $\tilde{\rho}$ in each iteration. 
We next explain how to implement the rest of the algorithm. 

\subsection{Implementation and time analysis} \label{sec:k-time}
We next explain how to implement each iteration. We assume during the algorithm that all the vertices know all the edges of $H$ and $A$. In our $k$-ECSS algorithm we build a $k$-edge-connected subgraph by applying the algorithm for $Aug_i$ for $i \leq k$. The input $H$ for $Aug_j$ is the set of edges added to the augmentation during the $j-1$ first times we run the algorithm for $Aug_i$. Hence, it is enough to show that at each run of the algorithm all the vertices learn about all the edges added to the augmentation. 

Given full  information about $H$ and $A$, each edge can compute how many cuts in $S_e$ are still not covered, which allows computing the cost-effectiveness in Line \ref{k-cost-ef}. Learning the maximum rounded cost-effectiveness in Line \ref{k-max} takes $O(D)$ rounds by a communicating over the BFS tree. To compute an MST in Line \ref{forest} we use the MST algorithm of Kutten and Peleg \cite{kutten1998fast} that takes $O(D + \sqrt{n} \log^*{n})$ rounds. Let $n_i$ be the number of edges added to the augmentation at iteration $i$. All the vertices learn these edges in $O(D + n_i)$ rounds by communication over the BFS tree.
Since all the vertices know the edges in $H$ and $A$ they can check if $H \cup A$ is $k$-edge-connected in Line \ref{k-term}, and detect the termination of the algorithm. 

The time complexity of one iteration is $O(D+ \sqrt{n} \log^*{n} + n_i)$. From Claim \ref{no-cycle}, the number of edges added to $A$ during the algorithm is at most $n-1$, which gives $\sum_i n_i < n$. Since the number of iterations is $O(\log^3{n})$, the time complexity of the algorithm is $O(D \log^3{n} + \sqrt{n} \log^3{n}\log^*{n} + \sum_i n_i) = O(D \log^3{n} + n)$, showing the following.

\begin{lemma} \label{time-i}
The time complexity of the algorithm is $O(D \log^3{n} + n)$, if all the vertices know the complete structure of $H$ at the beginning of the algorithm.
\end{lemma} 


\subsection{Approximation ratio analysis} \label{k-approx}

We next show that the approximation ratio of the algorithm is $O(\log{n})$ in expectation.
The general idea is similar to the analysis in our 2-ECSS algorithm. We assign a cost to all the cuts of size $k-1$ in $H$, as follows. 
For each cut $C$, we define $cost(C) = {1}/{\tilde{\rho}(e)}$ where $e$ is an edge that covers $C$ in the first iteration in which it is covered in the algorithm, and $cost(C) = 0$ if $\tilde{\rho}(e) = \infty$. Note that all the edges added to $A$ in a specific iteration have the same rounded cost-effectiveness, hence $cost(C)$ is well-defined.

Our goal is to show that $w(A) \approx \sum_{C} cost(C) \leq O(\log{n})w(A^*)$, where $A^*$ is an optimal augmentation.
The proof of the right inequality follows the proof of Lemma \ref{weighted} with slight changes, and is based on the fact that we always add edges with maximum rounded cost-effectiveness. 

For the left inequality, we show that $E[w(A)] \leq 6 \cdot E[\sum_C cost(C)]$. The proof is based on the following lemma. 
We define the \emph{degree} of a cut $C$ in an iteration to be the number of candidates that can cover it in this iteration, we denote it by $deg(C)$.

\begin{lemma} \label{deg} 
At the beginning of the phase where $p=\frac{1}{2^{\ell}}$, the maximal degree of a cut is at most $2^{\ell}$ with probability at least $1-\frac{1}{n^c}$ for a constant $c$. 
\end{lemma}

\begin{proof}
For $\ell = \lceil \log{m} \rceil$, it is clear. Consider the phase where $p=\frac{1}{2^{\ell}}$, 
we would like to show that at the end of the phase the degree decreases to $2^{\ell-1}$. Consider a specific cut $C$. Assume that $deg(C) > 2^{\ell-1}$ at the end of the phase (otherwise, we are done). 
Note that if at least one of the candidates that covers $C$ becomes an active candidate, then $C$ is covered at the end of the iteration by Claim \ref{cuts}.
Therefore, at each iteration where $deg(C) > 2^{\ell-1}$, the probability that $C$ does not get covered is $(1-\frac{1}{2^{\ell}})^{deg(C)} < (1-\frac{1}{2^{\ell}})^{2^{\ell-1}} \leq \frac{1}{\sqrt{e}}$. Hence, the probability that $C$ does not get covered in $t$ iterations of the phase is at most $(\frac{1}{\sqrt{e}})^t$. If we choose $t = 2c' \ln{n}$, we get $(\frac{1}{\sqrt{e}})^t = (\frac{1}{e})^{t/2} = \frac{1}{n^{c'}}$. The number of uncovered cuts of size $k-1$ is at most $n^2$, since there are at most $n \choose 2$ minimum cuts in a graph.
Using union bound, the probability that there is an uncovered cut with $deg(C) > 2^{\ell-1}$ at the end of the phase it at most $\frac{n^{2}}{n^{c'}}$. If we choose $c' = c + 2$, the probability is $\frac{1}{n^c}$. Hence, the probability that the maximal degree decreases to $2^{\ell-1}$ by the end of the phase is at least $1-\frac{1}{n^c}$. This completes the proof. 
\end{proof}

Our goal now is to show that $E[w(A)] \leq 6 \cdot E[\sum_{C} cost(C)].$ Our proof shows even slightly stronger claim: $E[w(A')] \leq 6 \cdot E[\sum_{C} cost(C)],$ where $A'$ are all the \emph{active candidates} during the algorithm (note that $A \subseteq A'$). This would be useful in our 3-ECSS algorithm. 

\begin{lemma}
$E[w(A)] \leq 6 \cdot E[\sum_{C} cost(C)].$
\end{lemma}

\begin{proof}
Let $A_j$ be the active candidates in iteration $j$, and let $CUT_j$ be the cuts that were first covered in iteration $j$. We show that $E[w(A_j)] \leq 6 \cdot E[\sum_{C \in CUT_j} cost(C)]$. Let $c_j$ be the maximal rounded cost-effectiveness in iteration $j$, then for all $C \in CUT_j$ it holds that $cost(C) = 1/c_j$. Hence, $\sum_{C \in CUT_j} cost(C) = \frac{1}{c_j} |CUT_j|$. On the other hand, 
$$w(A_j) = \sum_{e \in A_j} w(e) = \sum_{e \in A_j} |C_e| \cdot \frac{w(e)}{|C_e|} \leq \frac{2}{c_j} \sum_{e \in A_j} |C_e|,$$
where the last inequality follows from the fact that $\frac{w(e)}{|C_e|} = \frac{1}{\rho(e)} \leq \frac{2}{\tilde{\rho}(e)} = \frac{2}{c_j}$.
Hence, we need to show that $|CUT_j| \approx \sum_{e \in A_j} |C_e|$. Note that if each cut is covered exactly by one edge, then $|CUT_j| = \sum_{e \in A_j} |C_e|$. However, this is not necessarily true since several edges may cover the same cut in iteration $j$. We will show that $|CUT_j| \approx \sum_{e \in A_j} |C_e|$ in expectation. To do so, we write $\sum_{e \in A_j} |C_e| = \sum_C t(C)$ where we sum over all the cuts $C$ that are uncovered at the beginning of iteration $j$, and $t(C)$ is the number of active candidates that cover $C$ in iteration $j$.

We next estimate $E\big[\sum_{C} t(C) \big]$, the rest of the analysis is similar to the proof of Lemma 3.6 in \cite{jia2002efficient}. The main difference is that in \cite{jia2002efficient}, a condition similar to the condition described in Lemma \ref{deg} holds always, and in our case it holds w.h.p.

Let $p$ be the probability that a candidate edge is an active candidate in iteration $j$. If the maximal degree of an uncovered cut in iteration $j$ is at most $\frac{1}{p}$, we say that the iteration is \emph{good}. According to Lemma \ref{deg} it happens w.h.p.  
It holds that $$E[t(C)] = Pr[j\ is\ good] \cdot E[t(C)|\ j\ is\ good]+Pr[j\ is\ not\ good] \cdot E[t(C)|\ j\ is\ not\ good].$$
We start by analyzing $E[t(C)|\ j\ is\ good]$.
It holds that $$E[t(C)|\ j\ is\ good]=Pr[t(C)>0 | j\ is\ good] \cdot E[t(C)|t(C)>0, j\ is\ good] .$$ For an uncovered cut $C$, let $W$ be the set of candidates that can cover $C$ in iteration $j$. Each of them is an active candidate in iteration $j$ with probability $p$. If $j$ is good, then $deg(C) = |W| \leq \frac{1}{p}$, or equivalently $p \leq \frac{1}{|W|}$. 
Thus,
\begin{equation}
\begin{split}
E[t(C)|t(C)>0, j\ is\ good]=\sum_{e \in W} Pr[e\ is\ chosen|t(C) > 0,j\ is\ good] \\ = \sum_{e \in W} \frac{Pr[e\ is\ chosen, j\ is\ good]}{Pr[t(C) > 0, j\ is\ good]} = \frac{|W| \cdot p \cdot Pr[j\ is\ good]}{Pr[t(C) > 0 |j\ is\ good] \cdot Pr[j\ is\ good]}. 
\end{split}
\end{equation}

The second equality follows from the fact that if $e$ is chosen then necessarily $t(C)>0$. The third equality follows from the fact that the events ``$e\ is\ chosen$'' and ``$j\ is\ good$'' are independent.
We will show that $Pr[t(C) > 0 | j\ is\ good] \geq \frac{|W|p}{2}.$
It holds that $Pr[t(C) > 0] = 1-(1-p)^{|W|}$. Since $(1-p)^{|W|} \leq 1 - |W|p + {|W| \choose 2}p^2 \leq 1 - |W|p +\frac{(|W|p)^2}{2}$, we get $Pr[t(C) > 0 ] \geq |W|p - |W|p \cdot \frac{|W|p}{2}$. If $j$ is good, then $|W|p \leq 1$, which gives $Pr[t(C)>0 | j\ is\ good] \geq \frac{|W|p}{2}.$ 
This gives,
$$E[t(C)|t(C)>0, j\ is\ good] \leq \frac{|W|p}{|W|p/2} = 2.$$

This shows that $ E[t(C)|\ j\ is\ good] \leq 2 \cdot Pr[t(C)>0 | j\ is\ good].$
Summing over all the cuts $C$ that are not covered at the beginning of iteration $j$, we get
$$\sum_C E[t(C)|\ j\ is\ good] \leq 2 \cdot \sum_C Pr[t(C)>0 | j\ is\ good] = 2 \cdot E[|CUT_j| | j\ is\ good].$$

This gives,
$$\sum_C Pr[j\ is\ good] \cdot E[t(C)| j\ is\ good] \leq 2 \cdot Pr[j\ is\ good] \cdot E[|CUT_j| | j\ is\ good] \leq 2E[|CUT_j|].$$

By Lemma \ref{deg}, $j$ is not good with probability at most $\frac{1}{n^c}.$ Note that $t(C) \leq n^2$ always and there are at most $n^2$ cuts of size $k-1$, which gives
$$\sum_C Pr[j\ is\ not\ good] \cdot E[t(C) | j\ is\ not\ good] \leq \frac{n^4}{n^c}.$$
In addition, in each iteration there is at least one candidate $e$, that can cover at least one cut $C$. The cut $C$ is covered in iteration $j$ with probability at least $p \geq \frac{1}{m}$. This shows $E[|CUT_j|] = \sum_C Pr[t(C) > 0] \geq \frac{1}{n^2}.$ If we choose $c > 6$, we get
$$\sum_C Pr[j\ is\ not\ good] \cdot E[t(C) | j\ is\ not\ good]  \leq \frac{1}{n^2} \leq E[|CUT_j|].$$
To conclude, $E[\sum_C t(C)]= \sum_C E[t(C)]$ is equal to
$$\sum_C Pr[j\ is\ good] \cdot E[t(C)| j\ is\ good] + \sum_C Pr[j\ is\ not\ good] \cdot E[t(C) | j\ is\ not\ good] \leq 3E[|CUT_j|].$$
This shows $$E[w(A_j)] = \frac{2}{c_j} E[\sum_C t(C)] \leq \frac{6}{c_j} E[|CUT_j|] = 6 \cdot E[\sum_{C \in CUT_j} cost(C) ].$$ Summing over all iterations $j$ completes the proof. 
\end{proof}

To conclude, we get that $E[w(A)] \leq 6 \cdot E[\sum_{C} cost(C)] \leq O(\log{n})w(A^*)$, which shows an $O(\log{n})$ approximation in expectation. The time complexity of our algorithm for $Aug_i$ is $O(D \log^3{n} + n)$ by Lemma \ref{time-i}. To get a solution for $k$-ECSS we start with an empty subgraph $H$, and in iteration $i$ augment its connectivity from $i-1$ to $i$ using our algorithm for $Aug_i$. This guarantees that all the vertices learn the complete structure of $H$. 
Since we augment the connectivity $k$ times, by Claim \ref{Aug_k} we get the following.

\kECSS* 


\section{Unweighted 3-ECSS} \label{3-full}

In this section, we show how to improve the time complexity of our $k$-ECSS algorithm 
to $O(D \log^3{n})$ rounds for the special case of \emph{unweighted} 3-ECSS. 
The bottleneck of our $k$-ECSS algorithm is the cost-effectiveness computation. 
In the special case of \emph{unweighted} 3-ECSS we show how to compute it in $O(D)$ rounds. The main tool in our algorithm is the beautiful \emph{cycle space sampling} technique introduced by Pritchard and Thurimella \cite{pritchard2011fast}.

We say that $\{e,f\}$ is a \emph{cut pair} in a 2-edge-connected graph $G$ if removing $e$ and $f$ from $G$ disconnects it.
In \cite{pritchard2011fast}, it is shown how to assign each edge of a graph $G$, a short label $\phi(e)$ that allows to detect all the cut pairs of the graph, as follows. 
Two edges $e$ and $f$ are a cut pair if and only if $\phi(e) = \phi(f)$. The algorithm for computing the labels takes $O(D)$ rounds for a graph with diameter $D$. 

Our algorithm for unweighted 3-ECSS starts by computing a 2-edge-connected subgraph $H$, and then augments its connectivity. To build $H$, we use an $O(D)$-round 2-approximation algorithm for \emph{unweighted} 2-ECSS \cite{censor2017fast}. This algorithm starts by building a BFS tree $T$, and then augments its connectivity to 2. In particular, the diameter of $H$ is $O(D)$. For an edge $e \not \in T$, we define $S^1_e$ to be all the \emph{tree edges} in the unique tree path covered by $e$. We say that $e \not \in H$ covers a cut pair $\{f,f'\}$ if $\{f,f'\}$ is not a cut pair in $H \cup \{e\}$. 
For an edge $e \not \in H$, we define $S^2_e$ to be all the \emph{cut pairs} covered by $e$.
During the algorithm, we maintain a set of edges $A$ that includes all the edges added to the augmentation, initially $A = \emptyset$. $C_e$ are all the cut pairs in $S^2_e$ that are not covered by $A$.  
The cost-effectiveness of an edge $e \in H$ is defined as in $Aug_i$. However, since all the edges have unit-weight, $\rho(e) = |C_e|$.   

We next describe the algorithm for augmenting the connectivity of $H$ from 2 to 3.
The general structure of the algorithm is similar to our algorithm for $Aug_i$ in Section \ref{k-full},  
with slight changes. Now we just add all the active candidates to the augmentation without computing an MST, as follows. \\[-7pt]

\noindent\fbox{%
    \parbox{\textwidth}{%
    \vspace{-0.3cm}
\begin{enumerate}
\item { Each edge $e \not \in H \cup A$ computes its rounded cost-effectiveness $\tilde{\rho}(e)$.} \label{3-cost-ef}
\item { Each edge $e \not \in H \cup A$ with maximum rounded cost-effectiveness is a \textit{candidate}}.  \label{3-max} 
\item { Each candidate $e$ is added to the augmentation with probability $p_i$}. \label{3-prob}
\item { If all the cut pairs in $H$ are covered by $A$, the algorithm terminates, and the output is all the edges of $A$.} \label{3-term} 
\end{enumerate}
    \vspace{-0.4cm}
}} \\

The choice of the probability $p_i$ and the approximation ratio analysis follow our algorithm for $Aug_i$ and its analysis. 
Hence, our goal is to explain how to implement the algorithm efficiently. The computation in line \ref{3-max} takes $O(D)$ rounds as before, and the computation in line \ref{3-prob} is completely local. We next show how to implement lines \ref{3-cost-ef} and \ref{3-term} efficiently using the cycle space sampling technique.

\subsection{Overview of the cycle space sampling technique} \label{sec:cyc-sam}

The cycle space sampling technique allows to detect small cuts in a graph using connections between the cycles and cuts in a graph. In this section, we give a high-level overview of technique, for full details and proofs see \cite{pritchard2011fast}.

We say that a set of edges $\phi \subseteq E$ is a \emph{binary circulation} if all the vertices have even degree in $\phi$. For example, a cycle is a binary circulation. 
The \emph{cycle space} of a graph is the set of all the binary circulations. 
A binary circulation can be seen as a binary vector of length $|E|$, where each edge has an entry in the vector which is equal to 1 if and only if $e \in \phi$. Proposition 2.1. in \cite{pritchard2011fast} shows that the cycle space is a vector subspace of $\mathbb{Z}_2^E$, with the $\Moplus$ operation (where $\Moplus$ stands for addition of vectors modulo 2). Note that $\phi_1 \Moplus \phi_2$ is the symmetric difference of the sets $\phi_1,\phi_2$.
For $S \subseteq V $, denote by $\delta(S)$ the set of edges with exactly one endpoint in $S$. An \emph{induced edge cut} is a set of the form $\delta(S)$ for some $S$. In particular, $\{e,f\}$ is a cut pair in a 2-edge-connected graph if and only $\{e,f\}$ is an induced edge cut.
The following main lemma is essential for identifying the cuts in a graph (see Propositions 2.2 and 2.5 in \cite{pritchard2011fast}).

\begin{lemma} \label{cycle-sampling}
Let $\phi$ be a uniformly random binary circulation and $F \subseteq E$. Then
$$Pr[|\phi \cap F|\ is\ even] = \left\{
                \begin{array}{ll}
                  1,\ if\ F\ is\ an\ induced\ edge\ cut\\
                  1/2,\ otherwise
                \end{array}
              \right. $$ 
\end{lemma}

Hence, sampling a uniformly random binary circulation allows to detect if a set of edges is a cut with probability $1/2$. We next explain how to sample a binary circulation. Let $T$ be a spanning tree of a graph $G$ (in our algorithm we choose $T$ to be a BFS tree), and let $e$ be a non-tree edge, we denote by $Cyc_e$ the unique cycle in $T \cup \{e\}$. This is the \emph{fundamental cycle} of $e$. Note that $Cyc_e$ is composed of $S^1_e \cup \{e\}$ where $S^1_e$ are the edges in the unique tree path that $e$ covers. Proposition 2.3. in \cite{pritchard2011fast} shows the following.

\begin{claim} \label{basis}
The fundamental cycles of any spanning tree $T$ form a basis of the cycle space. 
\end{claim}

This gives a simple way to sample a binary circulation (see Proposition 2.6 and Theorem 2.7 in \cite{pritchard2011fast}). We choose a random subset $E'$ of non-tree edges, by adding each non-tree edge to $E'$ with independent probability $1/2$. Then, we define $\phi = \Moplus_{e \in E'} Cyc_{e}$. Since every non-tree $e$ edge is included only in the fundamental cycle $Cyc_e$, $\phi$ is composed of all the edges of $E'$ and all the tree edges that are included in odd number of the cycles $\{Cyc_e\}_{e \in E'}$. 

Lemma \ref{cycle-sampling} shows that sampling a random binary circulation allows to detect if a set of edges is a cut with probability $1/2$. To increase the success probability, we can sample many independent random circulations. A \emph{random b-bit circulation} is composed of $b$ mutually independent uniformly random binary circulations. Formally, let $\mathbb{Z}_2^b$ be the set of $b$-bit binary strings. For $\phi: E \rightarrow \mathbb{Z}_2^b$, let $\phi_i(e)$ denote the $i$th bit of $\phi(e)$. 

\begin{definition} 
$\phi: E \rightarrow \mathbb{Z}_2^b$ is a $b$-bit circulation if for each $1 \leq i \leq b$, $\{e | \phi_i(e) = 1\}$ is a binary circulation.
\end{definition} 

Corollary 2.9. in \cite{pritchard2011fast} follows from Lemma \ref{cycle-sampling}, and shows the following.

\begin{corollary} \label{cor-sampling}
Let $\phi$ be a random $b$-bit circulation and $F \subseteq E$. Then
$$Pr[\Moplus_{e \in F} \phi(e) = 0] = \left\{
                \begin{array}{ll}
                  1,\ if\ F\ is\ an\ induced\ edge\ cut\\
                  2^{-b},\ otherwise
                \end{array}
              \right. $$ 
Where $0$ is the all-zero vector. 
\end{corollary}

If we take $b = O(\log{n})$ for a sufficient large constant and focus on cut pairs, Corollary \ref{cor-sampling} shows that the following property holds w.h.p. 

\begin{property} \label{prop}
For all the edges, $\phi(e) = \phi(f)$ if and only if $\{e,f\}$ is a cut pair.
\end{property}

\begin{lemma} \label{whp}
Property \ref{prop} holds w.h.p.
\end{lemma}



Note that the error probability is only one-sided. If $\{e,f\}$ is a cut pair, then necessarily $\phi(e) = \phi(f)$. However, there is small probability that $\phi(e) = \phi(f)$ and $\{e,f\}$ is not a cut pair.

To detect cut pairs, we would like to sample a random $O(\log{n})$-bit circulation and let every edge $e$ learn the $O(\log{n})$-bit string $\phi(e)$. Since a random $O(\log{n})$-bit circulation is composed of $O(\log{n})$ independent random binary circulations, the following approach produces a random $O(\log{n})$-bit circulation.
Each non-tree edge $e$ chooses a uniformly independent $O(\log{n})$-bit string $\phi(e)$. This defines $\phi$ for all the non-tree edges, and can be computed locally by the non-tree edges. 
For a tree edge $t$, the label $\phi(t)$ is defined as follows: $\phi(t) = \Moplus_{t \in Cyc_e} \phi(e) = \Moplus_{t \in S^1_e} \phi(e)$. See Figure \ref{3-ECSS_pic} for an example.

\setlength{\intextsep}{0pt}
\begin{figure}[h]
\centering
\setlength{\abovecaptionskip}{-2pt}
\setlength{\belowcaptionskip}{6pt}
\includegraphics[scale=0.42]{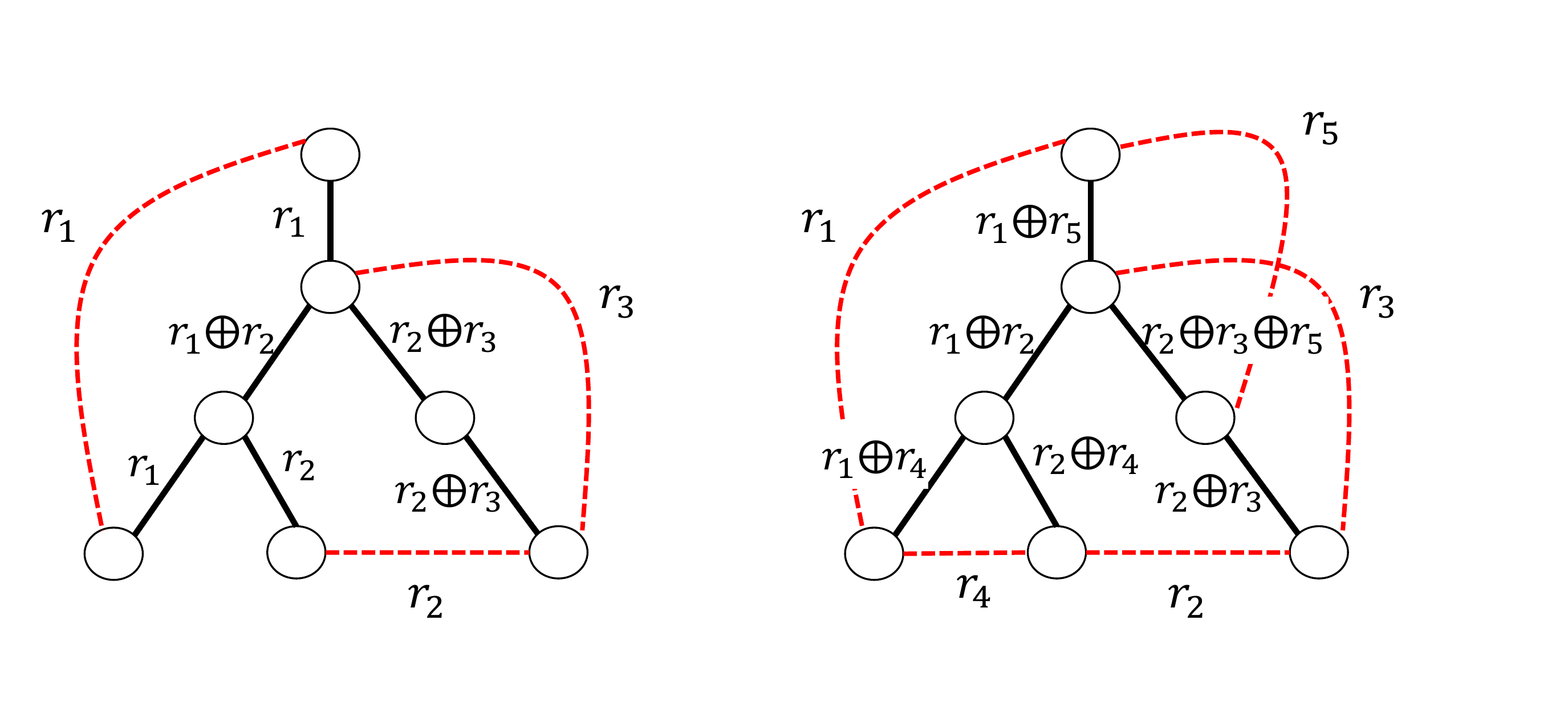}
 \caption{On the left there is a 2-edge-connected graph, with 3 non-tree edges with the labels $r_1,r_2,r_3$. The label of a tree edge $e$ is the xor of labels of non-tree edges that cover it. The cut pairs are all the pairs of edges with the same label. On the right, there are 2 additional non-tree edges, and now all the labels of edges are unique and there are no cut pairs.}
\label{3-ECSS_pic}
\end{figure}

In \cite{pritchard2011fast}, it is shown that the following algorithm allows to compute the labels of the tree edges in $O(D)$ rounds (see Theorem 4.2). 
We scan the tree from the leaves to the root, and each vertex $v$ computes the label of $\{v,p(v)\}$ according to the non-tree edges adjacent to it and the labels it receives from its children, as follows. $\phi(\{v,p(v)\})=\Moplus_{f \in \delta(v) \setminus \{v,p(v)\}} \phi(f)$, where $\delta(v)$ is the set of edges adjacent to $v$. The time complexity is $O(D)$ rounds for scanning the tree since it is a BFS tree, which gives the following. 

\begin{lemma} \label{alg-labels}
There is a distributed $O(D)$-round algorithm to sample a random $O(\log{n})$-bit circulation, where at the end each edge $e$ knows the value $\phi(e)$.
\end{lemma}

\subsection{Identifying cut pairs} \label{sec:cut-pairs}

We next give a simple characterization of all the cut pairs in a graph based on Section \ref{sec:cyc-sam}. Let $G$ be a 2-edge-connected graph $G$, and let $T$ be a spanning tree of it. Note that since $T$ is connected, any cut pair has at least one tree edge. We next show the following.


\begin{claim} \label{cut-pairs}
The edges $\{e,f\}$ are a cut pair in a 2-edge connected graph $G$ where $e$ is a tree edge, if and only if one of the following holds. 
\begin{enumerate}
\item $e$ is a tree edge and $f$ is the unique non-tree edge in $G$ that covers it.
\item $e$ and $f$ are tree edges covered by the exact same non-tree edges.
\end{enumerate}
\end{claim}

\begin{proof}
If $e$ is a tree edge and $f$ is a non-tree edge, then $\{e,f\}$ is a cut pair if and only if $f$ covers $e$ and it is the only non-tree edge that covers $e$, as otherwise there is another non-tree edge that covers $e$, and removing $e$ and $f$ from $G$ does not disconnect it.

We next consider the case that $e$ and $f$ are tree edges. Let $\phi$ be a random $b$-bit circulation.
Note that from the definition of labels, $\phi(t) = \Moplus_{t \in S^1_{e'}} \phi(e')$ for any tree edge $t$. This shows that two tree edges $e,f$ that are covered by the exact same non-tree edges have the same labels in any circulation $\phi$, which shows that $Pr[\phi(e)=\phi(f)]=1$. By Corollary \ref{cor-sampling},
it follows that $\{e,f\}$ is a cut pair in this case.

If $e$ and $f$ are tree edges that are not covered by the same non-tree edges, then there is some non-tree edge $g$ that covers exactly one of $e$ and $f$. Let $Cyc_g$ be the fundamental cycle of $g$, then it has exactly one of $e$ and $f$, which shows that $Cyc_g$ is a binary circulation where $|Cyc_g \cap \{e,f\}| = 1$. Hence, by Claim \ref{cycle-sampling}, $\{e,f\}$ is not a cut pair.
This completes the proof.
\end{proof}

The next corollary follows from Claim \ref{cut-pairs}.


\begin{corollary} \label{cor-pairs}
Let $H$ be a 2-edge-connected spanning subgraph of a graph $G$ and let $T$ be a spanning tree of $H$, an edge $e \not \in H$ covers the cut pair $\{f,f'\}$ if an only if exactly one of $f,f'$ is in $S^1_e$.
\end{corollary}

\begin{proof}
Let $\{f,f'\}$ be a cut pair in $H$.
If $f$ and $f'$ are two tree edges, they are covered by the exact same non-tree edges in $H$ according to Claim \ref{cut-pairs}. By the same claim, $\{f,f'\}$ is not a cut pair in $H \cup \{e\}$ if and only if they are not covered by the exact same non-tree edges in $H \cup \{e\}$. This happens if and only if $e$ covers exactly one of $f,f'$.

If $\{f,f' \}$ has one tree edge, say $f$, then $f'$ is the only non-tree edge that covers it in $H$ according to Claim \ref{cut-pairs}. By the same claim, $\{f,f'\}$ is not a cut pair in $H \cup \{e\}$ if and only if $f'$ is not the only non-tree edge in $H \cup \{e\}$ that covers $f$. This happens if and only if $e$ covers $f$.  
\end{proof}

\subsection{Implementing the algorithm} \label{sec:3-impl}

We next explain how to use the labels for computing the cost-effectiveness. First, we use the $O(D)$-round algorithm from Lemma \ref{alg-labels} to sample a random $O(\log{n})$-bit circulation $\phi$ of the 2-edge-connected graph $H \cup A$, and assign each edge $e \in H \cup A$ the label $\phi(e)$. 
By Lemma \ref{whp}, Property \ref{prop} holds w.h.p.
We next assume that this is the case, and show how to compute the cost-effectiveness. We later address the case that it does not happen, and show how it affects the algorithm. 

For computing cost-effectiveness, we need the following definitions. 
Let $t$ be a tree edge, we denote by $n_{\phi(t)}$ the number of edges in $H \cup A$ with the label $\phi(t)$. For an edge $e \not \in H \cup A$,  
we denote by $n_{\phi(t),e}$ the number of edges in $S^1_e$ with label $\phi(t)$. Note that if $\{f,f'\}$ is a cut pair, then $\phi(f) = \phi(f')$, we say that $\phi(f)$ is the label of the cut.
We next show the following.


\begin{claim} \label{cost-ef-3}
For an edge $e \not \in H \cup A$, the number of cut pairs with label $\phi(t)$ covered by $e$ is exactly $n_{\phi(t),e}(n_{\phi(t)} - n_{\phi(t),e})$.
\end{claim}

\begin{proof}
According to Corollary \ref{cor-pairs}, the edge $e$ covers a cut pair $\{f,f'\}$ if and only if exactly one of $f,f'$ is in $S^1_e$. Hence, for a tree edge $f \in S^1_e$ with label $\phi(t)$, the edge $e$ covers all the cut pairs of the form $\{f,f'\}$ where $\{f,f'\}$ is a cut pair and $f' \not \in S^1_e$, there are $n_{\phi(t)} - n_{\phi(t),e}$ cuts of this form. Since there are $n_{\phi(t),e}$ edges in $S^1_e$ with label $\phi(t)$, the number of cut pairs with label $\phi(t)$ covered by $e$ is exactly $n_{\phi(t),e}(n_{\phi(t)} - n_{\phi(t),e})$. 
\end{proof}

According to Claim \ref{cost-ef-3}, if we sum the value $n_{\phi(t),e}(n_{\phi(t)} - n_{\phi(t),e})$ over all the labels $\phi(t)$ where $n_{\phi(t),e} \geq 1$ we get the number of cut pairs in $H \cup A$ covered by $e$, which is exactly $\rho(e) = |C_e|$. Therefore, to compute cost-effectiveness we need to explain how each edge $e \not \in H \cup A$ learns all the relevant values $n_{\phi(t)},n_{\phi(t),e}$. The high-level idea is:\\[-7pt]

\noindent\fbox{%
    \parbox{\textwidth}{%
    \vspace{-0.3cm}
\begin{enumerate}[(a).]
\item Each non-tree edge $e \in H$ learns the values $(t,\phi(t))$ for all the edges in $S^1_e$. \label{step-1}
\item Each tree edge $t$ learns $n_{\phi(t)}$.\label{step-2}
\item Each edge $e \not \in H \cup A$ learns all the values $(t,\phi(t),n_{\phi(t)})$ for all the edges in $S^1_e$.\label{step-3} 
\item Each edge $e \not \in H \cup A$ deduces all the relevant values $n_{\phi(t)},n_{\phi(t),e}$. 
\end{enumerate}
	\vspace{-0.4cm}
}}\\

Note that after step (\ref{step-3}), each edge $e \not \in H \cup A$ knows all the labels of edges in $S^1_e$, and all the relevant values $n_{\phi(t)}$. This allows computing $n_{\phi(t)},n_{\phi(t),e}$.
We next explain how to implement steps (\ref{step-1})-(\ref{step-3}) in $O(D)$ rounds. 
Let $m_t$ be a piece of information of $O(\log{n})$ bits associated with the tree edge $t$. Then, all the vertices can learn in $O(D)$ rounds all the values $(t,m_t)$ for the tree edges in the tree path between them to the root $r$, as follows. First, each vertex $v$ sends to its children the message $(t,m_t)$ for the tree edge $t = \{v,p(v)\}$. Then, it sends to them the message it receives from its parent in the previous round. After $O(D)$ rounds each vertex $v$ knows all the values $(t,m_t)$ for all the tree edges in the path $P_{r,v}$.

Let $e=\{u,v\}$ be a non-tree edge, the vertices $v$ and $u$ can learn all the values $(t,m_t)$ for all the tree edges in $S^1_e$ by exchanging between them full information about the paths $P_{r,u},P_{r,v}$ in $O(D)$ rounds.  
This allows them to compute their LCA, learn which edges are in $S^1_e$ and learn all the relevant values $(t,m_t)$.
Hence, step (\ref{step-1}) can be easily computed in $O(D)$ rounds, and step (\ref{step-3}) can be computed in $O(D)$ rounds after all the tree edges learn the values $n_{\phi(t)}$. We next explain how the tree edges learn the values $n_{\phi(t)}$ in step (\ref{step-2}). 

Let $t$ be a tree edge with label $\phi(t)$. Since $H$ is 2-edge-connected, there is at least one non-tree edge $e \in H$ that covers $t$. Let $n_{\phi(t),e}$ be the number of tree edges in $Cyc_e = S^1_e \cup \{e\}$ with label $\phi(t)$. Then, the following holds.

\begin{claim}
If $e \in H$ covers $t$, then $n_{\phi(t),e} = n_{\phi(t)}$.
\end{claim}

\begin{proof}
According to Property \ref{prop} (that holds w.h.p by Lemma \ref{whp}), $\{f,f'\}$ is a cut pair in $H \cup A$ if and only if $\phi(f) = \phi(f')$. We consider two cases.

If $t$ is covered by only one non-tree edge in $H \cup A$, this edge is necessarily $e$, and $\phi(t) = \phi(e)$ by the definition of labels. Note that the edge $e$ is in a cut pair $\{e,f\}$ if and only if $f$ is a tree edge that is covered only by $e$ according to Claim \ref{cut-pairs}. Since $\{e,f\}$ is a cut pair if and only if $\phi(e)=\phi(f)$, it follows that all the edges with label $\phi(t)=\phi(e)$ are in $Cyc_e$, which shows that $n_{\phi(t),e} = n_{\phi(t)}$. 

If $t$ is covered by at least two edges in $H \cup A$, then $\{t,f\}$ is a cut pair if and only if $f$ is a tree edge covered by the exact same edges as $t$, according to Claim \ref{cut-pairs}. In particular, $e$ covers $f$. This shows that all the edges with label $\phi(t)$ are in $S^1_e \subseteq Cyc_e$, which gives $n_{\phi(t),e} = n_{\phi(t)}$. 
\end{proof}
 
Hence, to learn $n_{\phi(t)}$, it is enough to learn $n_{\phi(t),e}$ for an edge $e \in H$ that covers $t$.
In addition, the edge $e$ learns in step (\ref{step-1}) the labels of all the tree edges in $S^1_e$, which allows it computing $n_{\phi(t),e}$. Note that exactly one of the endpoints of $e$, say $u$, is a descendant of $t$ in $T$, and $u$ can pass the information $n_{\phi(t),e}$ to $t$. In order that all the tree edges would learn the values $n_{\phi(t)}$ simultaneously we use a pipelined upcast over the BFS tree. Each leaf $u$ sends to its parent the values $n_{\phi(t)}$ for all the tree edges above it in the tree where $u$ knows the value $n_{\phi(t)}$. Each internal vertex sends to its parent these values based on the messages it receives from its children, and the values already known to it. We can pipeline the computations to get a time complexity of $O(D)$ rounds. This completes the description of the cost-effectiveness computation, and shows how to implement Line \ref{3-cost-ef} of the algorithm in $O(D)$ rounds. 

To complete the description of the algorithm, we need to explain how to verify if $H \cup A$ is 3-edge-connected in Line \ref{3-term} of the algorithm in $O(D)$ rounds. To do so, we use again the algorithm for computing the labels. Now we apply it on the graph $H \cup A$ in Line \ref{3-term} (after we added new edges to $A$ in Line \ref{3-prob}). We also compute for each tree edge the value $n_{\phi(t)}$ in $O(D)$ rounds as before.


\begin{claim}
The graph $H \cup A$ is 3-edge connected if and only if $n_{\phi(t)}=1$ for all the tree edges.
\end{claim} 

\begin{proof}
According to Property \ref{prop}, $\{e,f\}$ is a cut pair if and only if $\phi(e)=\phi(f)$. 
Hence, $n_{\phi(t)}=1$ for all the tree edges, if and only if none of the tree edges is in a cut pair in $H \cup A$. Since any cut pair has at least one tree edge, this happens if and only if $H \cup A$ is 3-edge-connected.
\end{proof}

After each tree edge knows $n_{\phi(t)}$, we can check in $O(D)$ rounds if at least one of these values is greater than 1 by communication over the BFS tree. This completes the description of Line \ref{3-term} of the algorithm.

We next show that although our algorithm assumes that Property \ref{prop} holds, which happens w.h.p, the algorithm always terminates after $O(\log^3{n})$ iterations, at the end $H \cup A$ is \emph{guaranteed} to be 3-edge-connected, and the approximation ratio is still $O(\log{n})$ in expectation. 

\begin{lemma} \label{lem-ter}
The algorithm terminates after $O(\log^3{n})$ iterations, at the end $H \cup A$ is \emph{guaranteed} to be 3-edge-connected, and the approximation ratio is $O(\log{n})$ in expectation. 
\end{lemma}

\begin{proof}
According to Corollary \ref{cor-sampling}, if we choose $b = O(\log{n})$ for a sufficient large constant, with probability at least $\frac{1}{n^c}$ for all the edges in the graph, $\phi(e) = \phi(f)$ if and only if $\{e,f\}$ is a cut pair. In addition, the error probability is one-sided, if $\{e,f\}$ is a cut pair then necessarily $\phi(e)=\phi(f)$. 
Hence, if $n_{\phi(t)}=1$ for all tree edges, then the graph is necessarily 3-edge-connected. We also show the following.

\begin{claim} \label{ver-ter}
If $e \not \in H \cup A$ covers some cut pair $\{f,f'\}$, the value of cost-effectiveness computed by $e$ is at least 1. 
\end{claim}

\begin{proof}
If $e$ covers $\{f,f'\}$, then exactly one of $f,f'$, say $f$, is in $S^1_e$. The value $n_{\phi(f)}$ that $f$ and $e$ know is equal to $n_{\phi(f),e'}$ for some $e' \in H$ that covers $f$. It holds that $f' \in Cyc_{e'}$ (as otherwise $\{f,f'\}$ is not a cut pair in $H$). Since $f' \not \in S^1_e$, it must hold that $n_{\phi(f),e'} > n_{\phi(f),e}$. In addition, $n_{\phi(f),e} \geq 1$ since $f \in S^1_e$. This shows that $n_{\phi(f),e}(n_{\phi(f),e'} - n_{\phi(f),e})$ is at least 1, and hence the cost-effectiveness value computed by $e$ is at least 1.
\end{proof}


To address the possibility of error, we change the algorithm slightly, as follows. When computing the maximum cost-effectiveness in the graph we take the minimum value between the value computed in the current iteration and the last one. 
In addition, if at the previous iteration $p_i =1$ and the maximum rounded cost-effectiveness is $\rho$, we know that the maximum rounded cost-effectiveness must decrease, hence we take the minimum between $\frac{\rho}{2}$ and the value computed in the current iteration. As long as the graph $H \cup A$ is not 3-edge-connected, the maximum cost-effectiveness is at least 1 and at most $n^2$, which shows that there are at most $O(\log{n})$ values for the maximum rounded cost-effectiveness, and guarantees that the number of iterations is $O(\log^3{n})$ as before. The algorithm ends if either we find that the graph $H \cup A$ is 3-edge-connected in Line \ref{3-term} of the algorithm and then $H \cup A$ is necessarily 3-edge-connected, or after $O(\log^3{n})$ iterations. In the latter case, at the last iteration $p_i = 1$ and the maximum rounded cost-effectiveness is 2. Hence, all the edges that cover some cut pair in $H \cup A$ are added to the augmentation by Claim \ref{ver-ter}, which shows that at the end $H \cup A$ is 3-edge-connected.
Since the cost-effectiveness computations may be wrong with small probability, we may add to the augmentation edges that are not with maximum rounded cost-effectiveness. However, since the probability of error is small enough, we can show an approximation ratio of $O(\log{n})$ in expectation. This affects the right inequity in the approximation ratio analysis in Section \ref{k-approx}. Now, instead of showing $\sum_{C} cost(C) \leq O(\log{n})w(A^*)$, we show $E[\sum_{C} cost(C)] \leq O(\log{n})w(A^*)$, which still gives an approximation ratio of $O(\log{n})$ in expectation.
This completes the proof of Lemma \ref{lem-ter}.
\end{proof}

Let $H^*$ be an optimal solution for 3-ECSS. In our algorithm, we started with a 2-approximation for the minimum size 2-ECSS, $H$, and then augmented it to be 3-edge-connected. Since our algorithm gives an $O(\log{n})$-approximation in expectation for $Aug_3$, we get $E[|H \cup A|] \leq 2|H^*| + O(\log{n})|H^*| = O(\log{n})|H^*|$. Computing $H$ takes $O(D)$ rounds, and computing the augmentation $A$ takes $O(D \log^3{n})$ rounds, which shows the following.


\threeECSS*

\subsection{Remarks} \label{sec:remarks}

Our 3-ECSS algorithm works also for \emph{weighted} $3$-ECSS. However, in the weighted case our algorithm starts by computing an MST and not a BFS tree. Since the time complexity of the algorithm depends on the height of the tree, each iteration now takes $O(h_{MST})$ rounds instead of $O(D)$ rounds. Hence, the time complexity of the algorithm in the worst case is $O(n \log^3{n})$, which is worse than the algorithm in Section \ref{k-full}.  

In our 2-ECSS algorithm, each non-tree edge needed to learn an $O(\log{n})$-bit piece of information: how many tree edges vote for it or how many uncovered tree edges are in $S_e$. This allowed to parallelize the computations efficiently. Yet, in our 3-ECSS algorithm, non-tree edges need to learn \emph{all} the labels of the tree edges in $S^1_e$. For this reason, achieving a sublinear algorithm for weighted 3-ECSS seems to be more involved.

A key observation that allows us to achieve $O(D \log^3{n})$-round algorithm for unweighted 3-ECSS is that for any cut pair $\{f,f'\}$ in $H$ there is some non-tree edge $e \in H$ where $\{f,f'\} \subseteq Cyc_e$. For $k>3$ this observation is not true anymore, which suggests that this problem may be harder. 

\section{Discussion} \label{disc}

In this paper, we provide efficient distributed algorithms for $k$-ECSS. While our results improve significantly the time complexity of previous algorithms, many intriguing questions remain open.

First, our algorithms obtain $O(\log{n})$-approximations, and a natural question is whether it is possible to design efficient algorithms with a better approximation ratio. Our approach which is based on set cover allows us to parallelize the computations efficiently, however it cannot achieve an approximation better than $O(\log{n})$. Another option is to try to convert sequential algorithms for $k$-ECSS to distributed ones. However, algorithms that obtain constant approximations in the sequential setting seem inherently sequential \cite{goemans1994improved, jain2001factor, khuller1994biconnectivity}. 

Second, we have presented a sublinear algorithm for \emph{weighted} 2-ECSS. The $k$-ECSS problem seems to be more involved for $k>2$, and it would be interesting to study whether sublinear algorithms exist also for $k>2$, or alternatively prove that this problem is indeed harder. 
In addition, we showed here an $O(D \log^3{n})$-round algorithm for \emph{unweighted} 3-ECSS. A natural question is whether algorithms with a similar time complexity exist also for $k>3$. 

Finally, all our algorithms are randomized, and it would be interesting to study whether it is possible to obtain also \emph{deterministic} algorithms with a similar time complexity.

\paragraph*{Acknowledgments:} I would like to thank Keren Censor-Hillel for her guidance and support, and for many helpful comments. I am also grateful to Merav Parter for suggesting the connection between the FT-MST algorithm \cite{ghaffari2016near} and the 2-ECSS problem.

\bibliography{2-ECSS}
\bibliographystyle{plain} 

\end{document}